%% file: icml2021.tex
\documentclass{article}

\usepackage{microtype}
\usepackage{graphicx}
\usepackage{subfigure}
\usepackage{booktabs} 
\usepackage{bm}
\usepackage{amsfonts}
\usepackage{physics}
\usepackage{color}
\usepackage{tikz}
\usepackage{floatrow}

\usepackage{amsmath, amsthm, amssymb, mathtools}
\newtheorem{theorem}{Theorem}[section]
\newtheorem{lemma}[theorem]{Lemma}
\newtheorem{prop}[theorem]{Proposition}
\newtheorem{definition}{Definition}[section]

\newcommand{\id}{\mathbf{1}}
\newcommand{\Det}{\mathrm{Det}}
\newcommand{\e}{\mathrm{e}}

\usepackage{hyperref}

\usepackage[accepted]{arxiv}

\icmltitlerunning{The Hintons in your Neural Network:
a Quantum Field Theory View of Deep Learning}

\begin{document}

\twocolumn[
\icmltitle{The Hintons in your Neural Network:\\
a Quantum Field Theory View of Deep Learning}




\begin{icmlauthorlist}
\icmlauthor{Roberto Bondesan}{qualcomm}
\icmlauthor{Max Welling}{qualcomm}
\end{icmlauthorlist}

\icmlaffiliation{qualcomm}{Qualcomm AI Research, Qualcomm Technologies Netherlands B.V. (Qualcomm AI Research is an initiative of Qualcomm Technologies, Inc.)}

\icmlcorrespondingauthor{Roberto Bondesan}{rbondesa@qti.qualcomm.com}

\icmlkeywords{Machine Learning, ICML}

\vskip 0.3in
]



\printAffiliationsAndNotice{}  

\input{files/main}




\newpage
\bibliography{references}
\bibliographystyle{icml2021}

\newpage
\appendix
\include{files/app}

\end{document}

%% file: files/main.tex
\begin{abstract}
In this work we develop a quantum field theory formalism for deep learning, where input signals are encoded in Gaussian states, a generalization of Gaussian processes which encode the agent's uncertainty about the input signal. We show how to represent linear and non-linear layers as unitary quantum gates, and interpret the fundamental excitations of the quantum model as particles, dubbed ``Hintons''.
On top of opening a new perspective and techniques for studying neural networks, the quantum formulation is well suited for optical quantum computing, and provides quantum deformations of neural networks that can be run efficiently on those devices.
Finally, we discuss a semi-classical limit of the quantum deformed models which is amenable to classical simulation.
\end{abstract}


\section{Introduction and Contributions}
\label{sec:introduction}

Since its inception more than 10 years ago, deep learning has achieved some stunning successes, transforming entire fields such as speech recognition, automated translation, computer vision and protein folding. It's impact on industry and the economy is so large that dedicated chips are being developed just to process neural network workloads. A natural question arises, what will drive the next wave of innovation in this field?

One candidate for the next disruption in AI is quantum computing. Quantum computing is in its nascent phase with only a modest number of qubits or qumodes (i.e.~states of light) available to perform computations. Moreover, it is not at all clear how to map neural networks onto a quantum computer in way that provide clear benefits. Yet, the confluence of deep learning and quantum computing holds promise and deserves exploration.

In this paper we propose a direct mapping of a deep neural network onto a optical quantum computer through the language of quantum field theory. We consider data, e.g. an image, as a finite sampling of an underlying continuous signal. We model this using quantum optical Gaussian states and show that they represent a generalization of a Gaussian process. We show how to compute the posterior of this Quantum Gaussian Process given input data. Then we show how both linear layers and nonlinear layers can be represented as unitary transformations on these Gaussian states. In the process we define a new quantum nonlinearity analogous to the classical softplus.

Because nonlinearities map us outside of the scope of Gaussian states we discuss a number of special cases which are tractable to simulate classically. The full architecture can only be simulated on an optical quantum computer.

We also provide details on how to map our model onto an optical quantum computer, including quantum Gaussian state preparation, and the implementation of linear and nonlinear layers in terms of elementary universal gates.  

While we have implemented and tested our tractable approximations we emphasize that our
contribution is theoretical; we do not expect that these tractable approximations perform better than their classical cousins. By describing deep learning in the language of quantum field theory we pave the way for the development of novel quantum neural network architectures in the future. We find the match between neural networks and quantum field theory very natural indeed. And perhaps amusingly, in the process of formulating deep learning in the language of quantum field theory we discovered a new particle: The ``Hinton'' is the elementary excitation of the quantum field from which optical quantum neural networks are made.



\section{Related Work}
\label{sec:related_work}

Several works have recently investigated quantum neural networks as a  variational quantum circuit using qubit architectures, e.g.~\cite{farhi_neven,baqprop, beer,qcnn,huggins,bondesan2020quantum}.
Quantum optical neural networks similar to those described in this work have been discussed in \cite{killoran2019,steinbrecher2019quantum, PhysRevLett.118.080501, PhysRevA.97.022315}.
Optical devices have also been considered for implementations of classical neural networks on photonic hardware due to the efficiency of matrix multiplication with optical instruments in \cite{shen2017deep}.
While the implementation of linear layers considered in all these works is the same as ours, our work differs from previous literature in the following two main aspects: 1) we employ Gaussian states to reason about uncertainty due to discretization errors; 2) we show how to implement quantum analogs of popular nonlinearities such as softplus using unitary gates.
Technical details of these differences will be discussed after we introduce our framework below.
Using Gaussian states for interpolating data that is further processed by a neural network is inspired by the recent works \cite{li2016scalable,finzi2020probabilistic}, which however are concerned only with classical networks.

Quantum algorithms for Gaussian processes have been studied in \cite{zhao2019quantum, Das_2018}. However, the focus there is speeding up GP regression and not devising a quantum neural network. Further, none of these works discusses the connection between quantum optical Gaussian states and Gaussian processes.



\section{Background}
\label{sec:background}


\subsection{Probabilistic Numeric Neural Networks}
\label{sec:pncnn}

We start by reviewing the framework of probabilistic numeric neural networks for classifying an input signal with missing data \cite{finzi2020probabilistic}. 
This model
uses a GP on a continuous space ${\cal X}=\mathbb{R}^D$ to interpolate the input signal and defines a neural network on this GP representation. 
To simplify the discussion we consider here a finite dimensional input space. This also allows us to present a simple numerical procedure to study the models, and does not change the conceptual findings of the present paper.
Note that our approach is also related to \cite{li2016scalable} but we shall keep the term probability numeric to emphasize the connection with the field of probabilistic numerics \cite{cockayne2019bayesian}.

The input to the model is ${\cal D} = \{ (x_i,y_i) \}_{i=1}^N$ corresponding to the observation $y_i$ of a signal (field) $\varphi : {\cal X}\to \mathbb{R}$ at location $x_i\in {\cal X}$. 
We assume an underlying finite input grid ${\cal X}=\{1,\dots, l \}^{\times d}$ and denote by $\bm{\varphi}$ the vector with components $\varphi_x$. We consider a prior GP with zero mean and kernel $k$, and compute the posterior ${\cal GP}(\mu',k')$ to interpolate the signal \cite{rasmussen2006gaussian}:
\begin{equation}
\label{eq:posterior}
\mu_x'=\bm{k}_x^\top A^{-1} \bm{y}
\,,
\quad
k'_{x,x'}
=
k_{x, x'}
-
\bm{k}_x^\top A^{-1}
\bm{k}_{x'}\,.
\end{equation}
Here $\bm{k}_x=\{ k_{x,x_i} \}_{i=1}^N$,
$A_{i,j} = k_{x_i, x_j} + \delta_{i,j} \sigma_n^2$, where 
$\sigma_n^2$ is measurement noise.
We then apply a sequence of linear layers and nonlinearities on the input random vector  $\bm{\varphi}^{(1)}\sim {\cal GP}(\mu',k')$:
\begin{equation}\label{eq:layer_recurrence}
    \bm{\varphi}^{(\ell+1)} = 
    \sigma(
    \mathcal{B}^{(\ell)}
    (\mathcal{A}^{(\ell)}\bm{\varphi}^{(\ell)})).
\end{equation}
In the intermediate layers, $\bm{\varphi}$ has shape $N_C\times |{\cal X}|$, with $N_C$ being the number of channels, and components $\varphi_{a,x}$.
$\sigma$ a point-wise non-linearity, ${\cal B}$ adds a bias, and $\mathcal{A}$ is linear.
In the case of translation equivariance, \cite{finzi2020probabilistic} chose
\begin{align}
    \label{eq:cal_A}
    \mathcal{A} &= \sum_{k}W_k{\rm e}^{\mathcal{D}_k}\,,
\end{align}
with $[W_k]^{ab}$ a matrix of parameters, where $a,b$ index the channels, and
\begin{align}
{\cal D}_k
=
\sum_{i_1,\dots,i_d\ge 0} 
\alpha_{i_1,\dots,i_d}
\partial_1^{i_1}\cdots \partial_d^{i_d}
\,,
\end{align}
where, denoting the unit vector in direction $\mu\in\{1,\dots, d\}$ by $e_\mu$, 
$(\partial_\mu \varphi)_x = \varphi_{x+e_\mu} - \varphi_{x}$ and $\alpha_{i_1,\dots,i_d}$ is constant. 

After $L$ transformations of type \eqref{eq:layer_recurrence}, the output feature undergoes a global average pooling, $\bm{\varphi}^{(L+1)}={\cal P} \bm{\varphi}^{(L)}$, 
where 
\begin{align}
    \label{eq:cl_pooling}
    ({\cal P}\bm{\varphi})_{a,x}
    =
    \begin{cases}
    \frac{1}{|{\cal X}|}
    \sum_{x\in {\cal X}} \varphi_{a,x}
    &\text{ first element} \\
    \varphi_{a,x} &\text{ else}
    \end{cases}
    \,.
\end{align}
This implements an invertible version of the usual operation. 
Finally, the first element of $\bm{\varphi}^{(L+1)}$ is passed through a linear layer that produces an output whose mean is interpreted as logits for classification.
Denoting ${\cal L} = \mathcal{B} \circ \mathcal{A}$, 
the chain of operations of the PNCNN is
\begin{align}
    \label{eq:pncnn_map}
    \Phi = 
    {\cal L}^{(L)}
    \circ 
    {\cal P}\circ 
    \sigma \circ \mathcal{L}^{(L-1)}
    \circ \cdots 
    \sigma \circ \mathcal{L}^{(0)}
    \,.
\end{align}


\subsection{Quantum Mechanics}
\label{sec:quantum}

In this section we review some basics of quantum mechanics, in particular quantum fields. We refer the reader to e.g.~\cite{QM_book} for a standard introduction.
In the quantum formalism we associate to every classical configuration of a random field $\bm{\varphi} = \{ \varphi_x \}_{x\in {\cal X}} \in \mathbb{R}^{|{\cal X}|}$ (for simplicity we ignore the channels in this section) a vector $\ket{\bm{\varphi}}$. The span of all these vectors 
is a vector space ${\cal H}$, with elements given by superpositions:
\begin{align}
\ket{\Psi} = \int D(\bm{\varphi}) ~\psi(\bm{\varphi}) \ket{\bm{\varphi}}\,,
\quad D(\bm{\varphi})=\prod_{x\in{\cal X}}\dd \varphi_x
\,.
\end{align}
Note that $\psi(\bm{\varphi})$ are the coefficients $\in \mathbb{C}$ that are used to combine vectors $\ket{\bm{\varphi}}$. The bra-ket notation might feel a little strange to readers unfamiliar with quantum mechanics, and is used to represent abstract elements in ${\cal H}$. 

${\cal H}$ is a Hilbert space equipped with a scalar product 
$\braket{\bm{\varphi}}{\bm{\varphi}'}=\delta(\bm{\varphi}-\bm{\varphi}')$, so that $\bra{\bm{\varphi}} \ket{\Psi}= \psi(\bm{\varphi})$.
We call a linear operator on ${\cal H}$ a quantum field (indicated with a hat). The algebra of quantum fields is generated by the pairs $\{ \widehat{\varphi}_x, \widehat{\pi}_x \}_{x\in {\cal X}}$:
\begin{align}
\bra{\bm{\varphi}} \widehat{\varphi}_x \ket{\Psi}
&=
\varphi_x \psi(\bm{\varphi})
\,,\\
\bra{\bm{\varphi}} \widehat{\pi}_x \ket{\Psi}
&=
-i \frac{\partial}{\partial \varphi_x}
\psi(\bm{\varphi})\,.    
\end{align}
They are self-adjoint and satisfy the canonical commutation relations 
\begin{align}
\label{eq:ccr_phi_phi}
[\widehat{\varphi}_x, \widehat{\pi}_{x'}] := 
\widehat{\varphi}_x\widehat{\pi}_{x'}
-
\widehat{\pi}_{x'}\widehat{\varphi}_x
=
i \delta_{x,x'}\,.    
\end{align}
We will often need to compute analytic functions of operators $O$, e.g.~$f(O) = \e^{O}$. These expressions are defined by the Taylor expansion of $f$.

A quantum state is a normalized superposition: $\braket{\Psi}{\Psi}=1$.
We can define expectation values of an operator $O$ by its matrix elements in a state $\bra{\Psi} O \ket{\Psi}$. This reduces to classical expectation values when $O$ is diagonal:  $\bra{\Psi} O \ket{\Psi} =\int D(\bm{\varphi})D(\bm{\varphi'})\braket{\Psi}{\bm{\varphi}}\bra{\bm{\varphi}} O \ket{\bm{\varphi'}}\braket{\bm{\varphi'}}{\Psi} = \mathbb{E}_{\bm{\varphi}\sim |\psi(\bm{\varphi})|^2}[O(\bm{\varphi})]$.

Quantum dynamics needs to preserve the norm of quantum states and acts by unitary operators $\widehat{U}^t = \e^{-i t \widehat{H}}$, where $\widehat{H}$ is a self-adjoint Hamiltonian.
Instead of evolving states we can equivalently evolve the observables $\widehat{A}$ we are going to measure:  $\widehat{A}(t) = (\widehat{U}^t)^\dagger \widehat{A}\,\widehat{U}^t$, which can be computed using the Baker-Campbell-Hausdorff formula:
\begin{align}
\label{eq:bch}
\e^{+i t \widehat{H}} \widehat{A} \,
\e^{-i t \widehat{H}}
=
\widehat{A}
+it [\widehat{H}, \widehat{A} ]
- \frac{t^2}{2}
[\widehat{H}, [\widehat{H}, \widehat{A}]]
+\dots
\end{align}
$\widehat{A}(t)$ satisfies the Heisenberg equation of motion:
\begin{align}
\label{eq:eom}
    \frac{\dd \widehat{A}(t)}{\dd t} = 
    i[\widehat{H}, \widehat{A}(t)]
    \,.
\end{align}

Measurements reduce a quantum superposition to a classical configuration \footnote{One can consider measurements of non-diagonal self-adjoint fields $O$, which project onto an eigenstate of $O$, but in this work we shall restrict to measurements of diagonal operators.}. 
This projection $\ket{\Psi}\mapsto \ket{\bar{\bm{\varphi}}}$ occurs with probability $|\psi(\bar{\bm{\varphi}})|^2$.
We shall make use below of the following more general fact about partial measurements (see e.g.~\cite{watrous2018theory}).
\begin{prop}
\label{prop:partial_measurement}
Let $\ket{\Psi}$ be a prior state in ${\cal H}={\cal H}_1 \otimes {\cal H}_2$.
If a measurement on ${\cal H}_1$ gives outcome $y_1$, the posterior state is $\ket{\zeta_{y_1}}=
(P_{y_1}\otimes \id_2) \ket{\Psi}/
\sqrt{\bra{\Psi} P_{y_1}\otimes \id_2 \ket{\Psi}}$, with $P_y = \ket{y}\bra{y}$ the projector on state $\ket{y}$.
\end{prop}
We note that Bayes rule follows from proposition \ref{prop:partial_measurement}: a subsequent measurement of an observable on ${\cal H}_2$ with outcome $y_2$ on the state $\ket{\zeta_{y_1}}$ will give outcome probability:
\begin{align}
   p(y_2 | y_1)
   =
   \bra{\zeta_{y_1}}\id_1 \otimes P_{y_2} \ket{\zeta_{y_1}}
   =
   \frac{\bra{\Psi} P_{y_1}\otimes P_{y_2}\ket{\Psi}}
   {\bra{\Psi} P_{y_1}\otimes \id_2\ket{\Psi}}
   \,,
\end{align}
which coincides with 
$p(y_1, y_2) / p(y_1)$.

\subsection{Gaussian States}
\label{sec:gaussian_states}

Define the $2|{\cal X}|$ dimensional vector of operators:
\begin{align}
\widehat{\bm{R}}
=
(\widehat{\varphi}_1,\dots,\widehat{\varphi}_{|{\cal X}|},
\widehat{\pi}_1,\dots,\widehat{\pi}_{|{\cal X}|})
\,.
\end{align}
After introducing the symplectic form $J$, \eqref{eq:ccr_phi_phi} reads as:
\begin{align}
\label{eq:ccr}
[\widehat{R}_i, \widehat{R}_j]
=
i J_{ij}
\,,
\quad
J = 
\begin{pmatrix}
0 & \id_{|\cal X|}\\
-\id_{|\cal X|} & 0
\end{pmatrix}
\,.
\end{align}
Gaussian states are specified uniquely by their mean and covariance defined as:
\begin{align}
    \label{eq:mean}
    \bm{m} &= \bra{\Psi} \widehat{\bm{R}} \ket{\Psi}\\
    \label{eq:cov}
    \tfrac{1}{2}C_{ij} &= 
    \bra{\Psi} \tfrac{1}{2}
    (\widehat{R}_i\widehat{R}_j+\widehat{R}_j\widehat{R}_i) \ket{\Psi}
    -
    m_im_j\,.
\end{align}
Here and below we shall denote by $1,2$ the first and second $|{\cal X}|$ components related to $\widehat{\varphi}$, $\widehat{\pi}$ sectors:
\begin{align}
    \bm{m} = (\bm{m}^1, \bm{m}^2)
    \,,
    \quad
    C = \begin{pmatrix}
    C^{11} & C^{12}\\
    C^{21} & C^{22}
    \end{pmatrix}
    \,.
\end{align}

\begin{prop}
\label{prop:C}
The covariance matrix $C$ satisfies:
\begin{align}
    C=C^T\,,\quad
    C>0\,,\quad
    C + i J \ge 0\,.    
\end{align}
\end{prop}
See Appendix \ref{sec:proofs_sec_background} for a proof.
The condition $C + i J \ge 0$ encodes the uncertainty principle and distinguishes quantum Gaussian states from classical Gaussian distributions on phase space \cite{Bartlett_2012}. Figure \ref{fig:qblob} gives a visualization of this fact.

\begin{figure}
\floatbox[{\capbeside\thisfloatsetup{capbesideposition={right,top},capbesidewidth=4cm}}]{figure}[\FBwidth]
{\caption{The covariance ellipse $\frac{1}{2}
(\bm{z}-\bm{m})^T C (\bm{z}-\bm{m}) \le 1$, also known as ``quantum blob'' \cite{de2006symplectic}, in a 2d phase space  $\bm{z}=(\varphi,\pi)$. The area is proportional to $\Det(C)>1$ (uncertainty principle) and is preserved by \eqref{eq:D}, \eqref{eq:omega}.}\label{fig:qblob}}
{\input{figures/qblob}}
\vskip -0.2in
\end{figure}
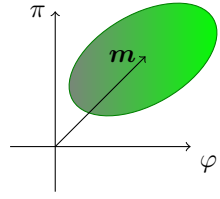

We denote a Gaussian state by $\ket{\bm{m}, C}$ \footnote{More general mixed Gaussian states can be defined, see e.g.~\cite{adesso2014continuous,Weedbrook_2012}, but we do not need such generality here.}. 
The wave-function $\braket{\bm{\varphi}}{\bm{m}, C}$ of a Gaussian state is a Gaussian function of $\bm{\varphi}$, albeit with complex quadratic form \cite{HUDSON1974249}.

We now show that the following unitary transformations, whose Hamiltonians are at most quadratic in $\widehat{R}_i$, implement the most general transformations among Gaussian states:
\begin{align}
    \label{eq:D}
    \widehat{D}(\bm{\xi})
    &=
    \e^{i \widehat{\bm{R}}^T J\bm{\xi}}
    \,,\quad \bm{\xi}\in \mathbb{R}^{2|{\cal X}|}
    \\
    \label{eq:omega}
    \widehat{\omega}(S)
    &=
    \e^{\frac{i}{2} \widehat{\bm{R}}^T JX 
    \widehat{\bm{R}}}
    \,,
    \quad 
    S=\e^X \in
    \text{Sp}_{2|{\cal X}|}(\mathbb{R})
    \,.
\end{align}
The Lie group $\text{Sp}_{2|{\cal X}|}(\mathbb{R})$ is the symplectic group of matrices satisfying $S J S^T = J$, which include rotations and scaling.
An element $X$ in its Lie algebra, $S=\e^X$, satisfies
$JX + X^TJ = JX - (JX)^T = 0$, 
so
\begin{align}
\left(
\widehat{\bm{R}}^T JX 
    \widehat{\bm{R}}\right)^\dagger
=
\widehat{\bm{R}}^T JX 
    \widehat{\bm{R}}
\,,
\end{align}
ensuring unitarity of $\widehat{\omega}(S)$.
$\widehat{D}$ and $\widehat{\omega}$ implement symmetry transformations (resp.~translations and linear symplectic transforms \footnote{More precisely, $\widehat{\omega}$ is a representation of the metaplectic group, a double cover of the symplectic group \cite{de2006symplectic}.}) on ${\cal H}$. The following proposition shows that $\widehat{\bm{R}}$ intertwines their action on ${\cal H}$ and the fundamental action on $\mathbb{R}^{2|{\cal X}|}$:
\begin{prop}
\label{prop:d_omega_trafo}
The unitaries of \eqref{eq:D} and \eqref{eq:omega} represent symplectic affine transformations of the canonical operators
\begin{align}
\widehat{D}(\bm{\xi})^\dagger
\widehat{\bm{R}}
\widehat{D}(\bm{\xi})
&=
\widehat{\bm{R}}
+
\bm{\xi}
\\
\widehat{\omega}(S)^\dagger
\widehat{\bm{R}}\,
\widehat{\omega}(S)
&=
S \widehat{\bm{R}}
\,.
\end{align}
\end{prop}
This result follows from formula \eqref{eq:bch}
and the commutation relations \eqref{eq:ccr}.
See Appendix \ref{sec:proofs_sec_background}
for details. This related proposition gives
the effect on the mean and covariance of Gaussian states and is
also proved in the Appendix \ref{sec:proofs_sec_background}.
\begin{prop}
\label{prop:gaussian_trafo}
Under the unitaries of \eqref{eq:D} and \eqref{eq:omega}, the Gaussian states transform as
\begin{align}
\widehat{D}(\bm{\xi})
\ket{\bm{m}, C}
&=
\ket{\bm{m}+\bm{\xi}, C}
\\
\widehat{\omega}(S)
\ket{\bm{m}, C}
&=
\ket{S\bm{m}, SCS^T}
\,.
\end{align}
\end{prop}


\section{Quantum Extensions of Probabilistic Numeric NNs}
\label{sec:qnn}

We introduce in this section a series of quantum operations that generalize the classical layers of a probabilistic numeric NN. 

\subsection{State Preparation}

We start by showing how to perform Bayesian inference with Gaussian states. From \ref{prop:partial_measurement} we have the following result.
\begin{lemma}[Quantum GP Inference]\label{lemma:qgp}
Let $\ket{0, C}$ be a Gaussian prior state such that
$C^{11}_{x,x'} = k_{x,x'}$.
Given data ${\cal D} = \{ (x_i,y_i) \}_{i=1}^N$ we consider the posterior:
\begin{align}
    \ket{\zeta_{\cal D}}
    =
    \frac{
    P_{\cal D}
    \ket{0, C}
    }
    {
    || P_{\cal D}
    \ket{0, C}||
    }
    \,,
    \quad
    P_{\cal D}
    = 
    \prod_{i=1}^N
    \ket{y_i}_{x_i}\prescript{}{x_i}{\bra{y_i}}\,,
\end{align}
where $\ket{y}_{x}$ is an eigenstate of $\widehat{\varphi}_x$.
We have:
\begin{align}
    \ket{\zeta_{\cal D}}
    =
    \ket{\bm{m}', C'}
    \,,
    \quad
    (m')^1_x = \mu_x'
    \,,\quad 
    (C')^{11}_{x,x'} = k'_{x,x'}
\end{align}
with $\mu',k'$ as in \eqref{eq:posterior}.
\end{lemma}
\begin{proof}
It follows directly from Prop.~\ref{prop:partial_measurement} with $\ket{\Psi} = \ket{0,C}$, ${\cal H}_1$ the space corresponding to locations $x_i$ and the formulas for the Gaussian conditionals \cite{rasmussen2006gaussian}.
\end{proof}
The quantum GP inference step allows one to encode a classical signal in a quantum state in such a way that quantum entanglement represents an agent's uncertainty about discretization errors.

\subsection{Quantum Linear Layers}

Next, we show how to perform the quantum equivalent of a linear layer that acts on the quantum fields $\widehat{\bm{R}}$ in the same way as a classical linear layer acts on a classical field 
$\bm{R}$. From Prop.~\ref{prop:d_omega_trafo} we have the following natural definition.
\begin{definition}\label{def:u_linear}
A quantum linear layer is the unitary:
\begin{align}
\widehat{U}_{\text{lin}}
(\bm{\xi}, S)
=
\widehat{D}(\bm{\xi})
\widehat{\omega}(S)
\,,
\end{align}
where $\widehat{D}(\bm{\xi})$ and $\widehat{\omega}(S)$ generalize the bias and multiplication by the weight matrix respectively.
\end{definition}

\subsection{Quantum Non-Linearity}
\label{sec:quantum_nl}

The definition of a unitary operator whose time evolution correspond to the action of a non-linearity is more involved, and in fact has remained elusive in the previous quantum neural network literature.
Similarly to classical non-linearities, a quantum non-linearity acts pointwise on the quantum fields. This restricts the associated Hamiltonian to be $\sum_{x,a} \widehat{H}_{x,a}$, where $\widehat{H}_{x,a}$ acts non-trivially only on the quantum fields at $x,a$. As a design principle, we consider the following class of time evolutions which map $\widehat{\varphi}_{x,a}$ to a function $\sigma(\widehat{\varphi}_{x,a})$.
(Recall that a function of an operator is defined by its Taylor series.)
This will provide a simple way to embed classical neural networks in our framework. 
\begin{prop}\label{prop:u_sigma}
Under the time evolution generated by 
\begin{align}
\label{eq:hf}
    \widehat{U}_{\sigma} &= \exp\Big(-i\sum_{x,a} \widehat{H}_{x,a}\Big)
    \,,\\
    \widehat{H}_{x,a} &= 
    \tfrac{1}{2}(
    \widehat{\pi}_{x,a} f(\widehat{\varphi}_{x,a})
    +
    f(\widehat{\varphi}_{x,a}) \widehat{\pi}_{x,a} 
    )\,,
\end{align}
the fields evolve according to the equations of motion:
\begin{align}
    \label{eq:eom_nl_phi}
    \dot{\widehat{\varphi}}_{x,a}(t)
    &=
    f(\widehat{\varphi}_{x,a}(t))
    \,,\\
    \label{eq:eom_nl_pi}
    \dot{\widehat{\pi}}_{x,a}(t)
    &=
    -
    \tfrac{1}{2}(
    \widehat{\pi}_{x,a}(t) f'(\widehat{\varphi}_{x,a}(t))
    +
    \text{h.c.}
    )
    \,,
\end{align}
where h.c. means the Hermitian conjugate of the expression preceding it.
\end{prop}
\begin{proof}
The equations of motion follow from \eqref{eq:eom} upon using the commutator $[\widehat{\pi}_{x,a}, f(\widehat{\varphi}_{x,a})] = -if(\widehat{\varphi}_{x,a})$.
\end{proof}

The following proposition relates $\sigma$ to $f$ (the proof checks the time derivatives and is in Appendix \ref{sec:proofs_qnn}).
\begin{prop}
\label{prop:ode_phi}
The ODEs \eqref{eq:eom_nl_phi}, \eqref{eq:eom_nl_pi} have solutions
\begin{align}
\widehat{\varphi}_{x,a}(t)
&=F^{-1}
\left(F(\widehat{\varphi}_{x,a}(0)) + t\right)
\,,\\
\widehat{\pi}_{x,a}(t)
&=
\frac{1}{2}
\left(
\widehat{\pi}_{x,a}(0)
\frac{f(\widehat{\varphi}_{x,a}(0))}
{f(\widehat{\varphi}_{x,a}(t))}
+\text{h.c.}\right)
\end{align}
where $F'(x) = 1/f(x)$.
\end{prop}
Prop.~\ref{prop:ode_phi} gives a general, albeit implicit, solution to the problem of constructing a quantum non-linearity. Next, we give an explicit solution for the case of softplus, a popular smooth version of ReLU.
% ...
\begin{lemma}[Quantum Softplus]
\label{lemma:q_softplus}
The softplus non-linearity with temperature $\beta$,
\begin{align}
\label{eq:softplus}
\sigma_\beta(x) = \frac{1}{\beta}\log(1+\e^{\beta x})   \,,
\end{align}
corresponds to time evolution from time $0$ to time $1$ under
\begin{align}
    \label{eq:h_softplus}
    \widehat{H}_{x,a}
    &= 
    \tfrac{1}{2\beta}
    (\widehat{\pi}_{a,x} \e^{-\beta\widehat{\varphi}_{a,x}}
    +
    \e^{-\beta\widehat{\varphi}_{a,x}} 
    \widehat{\pi}_{a,x})
    \,.
\end{align}
\end{lemma}
\begin{proof}
This Hamiltonian is of the general form \eqref{eq:hf} with $f(x) = \beta^{-1}
\e^{-\beta x}$. The antiderivative of its inverse, 
$\beta \e^{+\beta x}$, is $F(x) = \e^{\beta x}$. So $F^{-1}(y) = \beta^{-1}\log(y)$ and
the result then follows immediately from Prop.~\ref{prop:ode_phi}.
\end{proof}
We now move on to discussing architectures and symmetries of quantum neural networks. We shall resume the study of non-linearities in section \ref{sec:classical_simulation} where we will look at a semi-classical limit of the quantum models.

\subsection{Architecture}
\label{sec:q_arch}

We define a quantum neural network by putting these pieces together:
\begin{align}
    \label{eq:uqnn}
    \widehat{U}_{\text{NN}}
    =
    \widehat{U}_{\text{lin}}
    (\bm{\xi}^{(L)}, S^{(L)})
    \widehat{U}_{{\cal P}}
    \prod_{\ell=0}^{L-1}
    \widehat{U}_{\sigma}
    \widehat{U}_{\text{lin}}
    (\bm{\xi}^{(\ell)}, S^{(\ell)}) 
    \,,
\end{align}
where $\widehat{U}_{\sigma}$ and $\widehat{U}_{\text{lin}}$ are as in 
Prop.~\ref{prop:u_sigma} and Def.~\ref{def:u_linear}, while
\begin{align}
\widehat{U}_{{\cal P}} = \widehat{\omega}(S = {\cal P}\oplus ({\cal P}^{-1})^T)\,,
\end{align}
is a global average pooling operator, where ${\cal P}$ is as in \eqref{eq:cl_pooling}.
To make a prediction we proceed similarly to \cite{finzi2020probabilistic}. We discard the spatial locations that have not been aggregated over by averaging, act with a final linear classifier and finally measure the means $\widehat{\varphi}_{c}$ for the $c=1,\dots,C$ channels, $C$ being the number of classes:
\begin{align}
\label{eq:q_logits}
l_c = 
\bra{\zeta_{\cal D}}
\widehat{U}_{\text{NN}}^\dagger \widehat{\varphi}_{c}  \widehat{U}_{\text{NN}} \ket{\zeta_{\cal D}}
\,.
\end{align}
These are interpreted as the logits for the classification task at hand.
For simplicity we introduced the state $\ket{\zeta_{\cal D}}$ in lemma \ref{lemma:qgp} without referring to the channel dimension. To make sense of the action of $\widehat{U}_{\text{NN}}$ on it, we add extra registers for the channel dimension which are initialized to the vacuum state, which is a Gaussian state with zero mean and unit covariance \cite{adesso2014continuous}.
While we considered here only a global average pooling at the end following \cite{finzi2020probabilistic}, it is possible to pool features in intermediate layers as well by simply discarding registers associated to the modes to be discarded.
Computing the logits of \eqref{eq:q_logits} is in general intractable classically and requires a quantum computer. 

Our definition of a quantum neural network is similar to that of \cite{killoran2019}. However, w.r.t.~that work we introduce the following two main novelties: 1) 
we use Gaussian states for data interpolation (lemma \ref{lemma:qgp}); 2) we use a unitary gate implementing the non-linearity (see Sec.~\ref{sec:quantum_nl}). \cite{killoran2019} discusses two sets of non-linearities: the first uses a quantum channel to implement a non-linearity of the type $\widehat{\varphi}_{x,a}\mapsto \sigma(\widehat{\varphi}_{x,a})$, but, w.r.t.~our setting, that implementation requires to double the number of quantum registers and then discard half of them.\footnote{In our setting $\sigma$ is constrained to be an ODE flow. However, this is not restrictive from the point of view of expressivity of the neural network due to the presence of linear layers.}
The second strategy is to employ Hamiltonians such as
$
    \widehat{H}_{\text{cubic}}
    =
    \widehat{\varphi}_{a,x}^3
    $
    or
    $
    \widehat{H}_{\text{Kerr}}
    =
    (\widehat{\varphi}_{a,x}^2 + \widehat{\pi}_{a,x}^2)^2
$
as non-linearities, since either of them, together with the unitaries of \eqref{eq:D} and \eqref{eq:omega}, form a simple set of universal gates for quantum computation with continuous variables \cite{Lloyd_1999}, i.~e.~any other quantum gate can be expressed in terms of those.\footnote{
The quantum linear layers together with $\widehat{U}_\sigma$
also provide a universal set if $f$ is at least a polynomial of degree $3$.
}
While requiring minimal resources, these choices are likely to perform worse than our choice. Indeed in the semi-classical limit discussed in section \ref{sec:classical_simulation}, these correspond to low degree polynomial non-linearities, which are not efficient for classical neural network approximation  \cite{pinkus_1999}.


\subsection{Symmetries}

Generally, symmetries in classical neural networks are realized as linear maps $g\in G$ that acts on the activations $\bm{\varphi}$ as $\rho(g) \bm{\varphi}$, where $\rho$ is a representation matrix.
On top of translations, 
prominent examples of $G$ in ML are rotations \cite{cohen2018spherical} and permutations \cite{maron2018invariant}.
Having replaced the linear action on activations with $\widehat{\omega}$, we define unitary representations of $G$ on quantum states by $\widehat{\omega}(S_g := \rho(g) \oplus \rho^*(g))$, where $\rho^*(g) = \rho(g^{-1})^T$ is the dual representation, ensuring symplecticity of $S_g$. We have
\begin{align}
\widehat{\omega}(S_g)^\dagger    
\begin{pmatrix}
\widehat{\bm{\varphi}}
\\
\widehat{\bm{\pi}}
\end{pmatrix}
\widehat{\omega}(S_g)
=
\begin{pmatrix}
\rho(g)\widehat{\bm{\varphi}}
\\
\rho^*(g)\widehat{\bm{\pi}}
\end{pmatrix}
\,.
\end{align}
For example, in case of translations along the $\mu\in\{1,\dots, d\}$ axis,
$(\tau_\mu\bm{\varphi})_{a,x} = \varphi_{a,x+e_\mu}$,
we have $S_{\tau_\mu} = \tau_\mu\oplus \tau_\mu$, which translates both 
$\widehat{\bm{\varphi}}, \widehat{\bm{\pi}}$ variables equally.

Equivariance of a quantum linear layer $\widehat{\omega}(S)$ now amounts to the commutation relations:
\begin{align}
\widehat{\omega}(S)
\widehat{\omega}(S_g)
=
\widehat{\omega}(S_g)
\widehat{\omega}(S)
\Rightarrow
S S_g = S_g S
\end{align}
where the second formula follows from the group homomorphism property: $\widehat{\omega}(S)\widehat{\omega}(S')=\widehat{\omega}(SS')$. 
The characterization of symmetries presented here completely solves the problem of designing equivariant quantum linear layers by reducing the problem to designing equivariant classical linear  layers with symplectic weight matrices $S$ which commute with $\rho(g)\oplus \rho^*(g)$.
As in the classical case, since the non-linearities act pointwise, they will be invariant under operations that permute the coordinates $(\widehat{\varphi}_{x,a}, \widehat{\pi}_{x,a}) \mapsto (\widehat{\varphi}_{x',a'}, \widehat{\pi}_{x',a'})$, such as spatial symmetries, ensuring equivariance of the whole architecture.

As an illustration, the condition $SS_{\tau_\mu}=S_{\tau_\mu}S$ restricts each $M\times M$ block of 
\begin{align}
    S
    =
    \begin{pmatrix}
        A & B \\
        C & D
    \end{pmatrix}
    \,,
\end{align}
to be a convolution. Similarly, for rotations $G= \text{SO}(3)$, irreps are self-dual, so $S_{g} = \rho(g) \oplus \rho(g)$ and each block of $S$ is a group convolution \cite{cohen2018spherical}.



\section{Hintons}
\label{sec:hintons}

We discuss now a particle interpretation of the formalism introduced for neural networks.
We introduce the operators:
\begin{align}
    \widehat{\bm{b}} 
    = 
    \frac{1}{\sqrt{2}}
    (\widehat{\bm{\varphi}} + i\widehat{\bm{\pi}})
    \,,\quad
    \widehat{\bm{b}}^\dagger
    = 
    \frac{1}{\sqrt{2}}
    (\widehat{\bm{\varphi}} - i\widehat{\bm{\pi}})
    \,.
\end{align}
We can then use a particle basis (a.k.a. Fock space) for ${\cal H}$, where we identify a zero mean and unit covariance Gaussian state with no particles (vacuum)  $\ket{\Omega} = \ket{0,\id}$ such that $\widehat{\bm{b}} \ket{\Omega} = 0$ and create an orthogonal basis by acting with different monomials $(\widehat{b}_{a_1,x_1}^\dagger)^{n_i}\cdots (\widehat{b}_{a_m,x_m}^\dagger)^{n_m}$ on $\ket{\Omega}$, $n_i$ being the number of particles at channel $a_i$ and location $x_i$.
In the quantum optical setting, the particles are called photons. In the neural network context we dub these fundamental excitations ``Hintons'' after G.~Hinton, a founding father of the field of deep learning.

In the next sections, we will analyze tractable limits of the quantum neural network defined so far to get more insights into the architectural design.


\section{The Case of Classical Probabilistic NNs}
\label{sec:qft_repr}

Here we show how a classical probabilistic NN can be embedded in the quantum model of Sec.~\ref{sec:q_arch}.
First we prove the following quantum representation of the push forward of a GP under a generic classical (invertible) map (proof in the Appendix \ref{sec:proofs_sec_qft_repr}). 
\begin{lemma}\label{lemma:push_forward_generic}
Let $\ket{\bm{m},C}$ be a Gaussian state and $\widehat{U}$ a unitary such that:
\begin{align}
    \label{eq:U_F}
    \widehat{U}^\dagger
    \widehat{\bm{\varphi}}
    \widehat{U}
    =
    F(\widehat{\bm{\varphi}})
    \,.
\end{align}
Then:
\begin{align}
    |\bra{\bm{\varphi}}
    \widehat{U}
    \ket{\bm{m}, C}
    |^2
    =
    \left(F \# {\cal GP}(\bm{m}^1, C^{11})\right)(\bm{\varphi})
\end{align}
where $f \# p$ denotes the push forward of $p$ under $f$.
\end{lemma}

We already know that $\widehat{U}_\sigma$ has this property, so we only need to constrain the quantum linear layers so that they do not mix $\widehat{\bm{\varphi}}$ with $\widehat{\bm{\pi}}$. 
\begin{theorem}\label{thm:qft_repr_pncnn}
Consider the quantum network \eqref{eq:uqnn} with:
\begin{align}
    S^{(\ell)}
    =
    \begin{pmatrix}
    {\cal A}^{(\ell)} & 0 \\
    0 & (({\cal A}^{(\ell)})^{-1})^T
    \end{pmatrix}
    \,,\quad
    \bm{\xi}^{(\ell)} = (\bm{b}^{(\ell)}, 0)
    \,.
\end{align}
We have the quantum--classical duality:
\begin{equation}
    |\bra{\bm{\varphi}} \widehat{U}_{\text{NN}} 
    \ket{\zeta_{\cal D}}  |^2
    =
    (\Phi \# {\cal GP}(\mu',k'))(\bm{\varphi})
    \,,
\end{equation}
where the rhs is the push forward of the GP posterior \eqref{eq:posterior} under the map $\Phi$ of \eqref{eq:pncnn_map}.
That is, the logits of Eq.~\eqref{eq:q_logits} computed by the quantum neural network coincide exactly with those computed by the probabilistic numeric NN of section \ref{sec:pncnn} with weights and biases ${\cal A}^{(\ell)}, \bm{\xi}^{(\ell)}$.
\end{theorem}
\begin{proof}
Recalling the linear transformation law of Prop.~\ref{prop:d_omega_trafo} and the non-linear action of $\widehat{U}_\sigma$, the proof then follows from lemma \ref{lemma:push_forward_generic} with $F=\Phi$.
\end{proof}



\section{The Semiclassical Limit}
\label{sec:classical_simulation}

Let us denote by  $\bm{R} = (\bm{\varphi}, \bm{\pi})$ the classical fields corresponding to the quantum operators introduced above. Recall that notationally we distinguish between classical and quantum fields by the absence or presence of a hat.

Note that despite producing entangled states, the quantum linear layers acting on Gaussian states can be simulated efficiently on a classical computer as the action amounts to the matrix multiplications of Prop.~\ref{prop:d_omega_trafo}. See also \cite{Bartlett_2002}. In fact, at the linear level the only difference between a quantum evolution and a probabilistic classical evolution of a Gaussian Liouville measure in phase space is the covariance condition $C+iJ\ge0$ (Prop.~\ref{prop:C}) coming from the non-commutativity of position and momenta in quantum mechanics. We remark also that this condition is preserved by classical evolution thanks to the symplectic nature of classical mechanics \cite{de2009symplectic}.
In Sec.~\ref{sec:qft_repr} we showed how restricting the linear layers to block diagonal matrices led to the classical model, which corresponds to the push forward of an initial Gaussian Liouville distribution under a neural network, but this time only involving the $\bm{\varphi}$ field.
In this section, as an intermediate step towards studying the full quantum model, we change the non-linearity in such a way that the modified model corresponds to the push forward of a initial Gaussian measure under a neural network, involving both the $\bm{\varphi}$ and
the $\bm{\pi}$ fields.
We can interpret the resulting model as a semi-classical limit of the quantum model since it uses elements of quantum mechanics (uncertainty relation for the covariance) as well as classical mechanics (for the non-linearity).
Conceptually, we have the inclusion of models as special cases depicted in figure \ref{fig:hierarchy}.

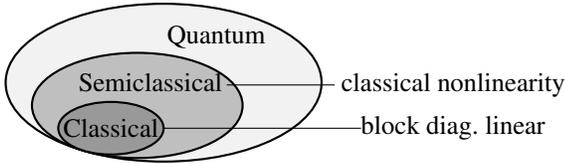
\begin{figure}[ht]
\vskip -0.1in
\begin{center}
\input{figures/hierarchy}
\caption{
Hierarchy of the neural networks considered.
}
\label{fig:hierarchy}
\end{center}
\vskip -0.2in
\end{figure}

We define the semiclassical model by simply replacing $\widehat{U}_\sigma$ with a classical Hamiltonian evolution under which the phase space measure evolves into a new classical phase space measure, which corresponds to the classical limit of the equation of motion \cite{QM_book}.

\begin{prop}\label{prop:u_sigma_cl}
Under the classical time evolution generated by the Hamiltonian
\begin{align}
\label{eq:hf_cl}
    H &= 
    \sum_{x,a}\pi_{x,a} f(\varphi_{x,a})
    \,,
\end{align}
the fields transform as:
\begin{align}
\varphi_{x,a}(t)
&=F^{-1}
\left(F(\varphi_{x,a}(0)) + t\right)
\,,\\
\pi_{x,a}(t)
&=
\pi_{x,a}(0)
\frac{f(\varphi_{x,a}(0))}
{f(\varphi_{x,a}(t))}
\end{align}
where $F'(x) = 1/f(x)$.
\end{prop}
\begin{proof}
The proof follows by simply checking the classical equations of motion:
\begin{align*}
    \dot{\varphi}_{x,a}(t)
    =
    f(\varphi_{x,a}(t))
    \,,\quad
    \dot{\pi}_{x,a}(t)
    =
    -
    \pi_{x,a}(t) f'(\varphi_{x,a}(t))
\end{align*}
\end{proof}

We note that classical and quantum equations of motions and solutions (Prop.~\ref{prop:ode_phi}) look identical. This is a consequence of the correspondence between quantum and classical mechanics under the identification 
$[\widehat{A},\widehat{H}] \longleftrightarrow i\hbar\{A,H\}$, where $\{\cdot, \cdot\}$ is the Poisson bracket.

In particular the following is the classical counterpart of lemma \ref{lemma:q_softplus}.
\begin{prop}[Symplectic Softplus]\label{lemma:s_softplus}
Replacing operators with classical variables 
in lemma \ref{lemma:q_softplus} we have
\begin{align}
    \label{eq:u_sigma_cl}
    U_\sigma 
    :
    \begin{pmatrix}
    \bm{\varphi}
    \\
    \bm{\pi}
    \end{pmatrix}
    \mapsto
    \begin{pmatrix}
    \frac{1}{\beta}\log(1 + \e^{\beta\bm{\varphi}})
    \\
    \bm{\pi}
    (1 + \e^{-\beta\bm{\varphi}})
    \end{pmatrix}
    \,.
\end{align}
\end{prop}

We then define a neural network that pushes forward the input Gaussian Liouville distribution on phase space to an output distribution $p_{\text{out}}$ by alternating linear layers with non-linear classical layers $U_\sigma$. Analogously to the original construction of \eqref{eq:q_logits}, its mean is then used as the logits for classification:
\begin{align}
    l_c
    =
    \mathbb{E}_{\varphi_c\sim p_{\text{out}}}(\varphi_c)
    \,.
\end{align}

The semiclassical neural network is a type of learnable Hamiltonian flow \cite{bondesan2019learning,rezende2019equivariant,toth2020hamiltonian}.
Adding momenta can be interpreted as an augmentation strategy for neural ODEs \cite{dupont2019augmented,massaroli2021dissecting}.
However we do not expect that the semiclassical neural network outperforms the classical ones of \cite{li2016scalable,finzi2020probabilistic}, since already in the classical case we avoid the expressivity restrictions of ODEs thanks to the presence of extra channels.

We implemented and tested the semiclassical neural network with symplectic softplus non-linearity.
The experiments are performed on a simple classification task for irregularly sampled time series. 
As already remarked, we do not claim that the semi-classical network performs better than a classical probabilistic numeric neural network with similar capacity. We therefore emphasize that the purpose of our experiments is simply to check that after adding momenta and using a different non-linearity than standard, the model performs similarly to a classical baseline.
In the implementation we saw that taking $\beta=0.1$ for the non-linearity avoids numerical instabilities for negative $\bm{\varphi}$ due to the presence of $\e^{-\beta \bm{\varphi}}$ in \eqref{eq:u_sigma_cl}.
We refer to Appendix \ref{sec:experiments} for details of architecture, task and learning algorithm tested.
We leave the study of the fully quantum case as an outstanding open problem which will most likely be tackled only when error corrected quantum optical computers will become available.




\begin{figure*}[ht]
\begin{center}
\centerline{\includegraphics[width=\textwidth]{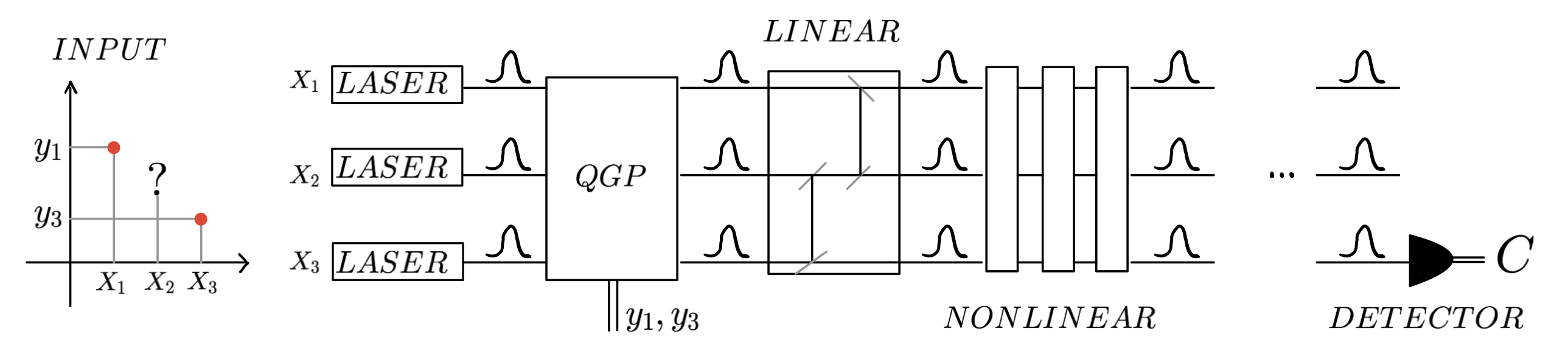}}
\caption{
High level depiction of the implementation of our model in quantum optical hardware.
The input on the left are observations $y_1, y_3$ of a signal at locations $x_1, x_3$ -- information at the intermediate value $x_2$ is missing. We then prepare laser beams for all locations $x_1,x_2,x_3$ and use quantum GP (QGP) inference to create a posterior state. We then apply a series of linear and non-linear layers, till we measure an observable with a detector to get a class $C$ for classifying the input signal.
}
\label{fig:pipeline}
\end{center}
\vskip -0.2in
\end{figure*}

\section{Quantum Optical Implementation}
\label{sec:quantum_optics}

We here describe practical considerations for implementing our model on an optical quantum computer.
Optical devices are attractive since one can implement matrix multiplication in a very efficient way \cite{shen2017deep} and quantum optical implementations of quantum neural networks have been discussed in \cite{killoran2019,steinbrecher2019quantum, PhysRevLett.118.080501, PhysRevA.97.022315}.
A graphical depiction of our proposal is in Fig.~\ref{fig:pipeline}.
While the implementation of linear layers is common to these works, the quantum non-linearities and the data embedding in Gaussian states are original to our work.

\subsection{State preparation}
\label{sec:state_prep}

The state preparation step encodes the input signal ${\cal D}=\{(x_i, y_i)\}_{i=1}^N$ into quantum registers. 
We compute the GP posterior as in \eqref{eq:posterior} for a set of points $x\in {\cal X}$ and then we create an input Gaussian state by acting on the vacuum state $\ket{\Omega}$ defined in Sec.~\ref{sec:hintons} with the linear layer $\widehat{D}(\bm{\xi}=(\bm{\mu}', 0))\widehat{\omega}(S=A\oplus A^{-1})$, where $A$ is a square root of $k'$.
The vacuum state can be created using lasers \cite{NielsenChuang} and we defer to the next section the implementation of the linear layer. 
This procedure incurs a complexity similar to the classical GP inference, that is $O(N^3)$.
In Appendix \ref{sec:state_prep_app} we give further comments on the state preparation step.
We see developing a more efficient quantum implementation of quantum GP inference as in interesting future direction.


\subsection{Linear layer}
\label{sec:linear_layer_impl}

The implementation of a unitary $\widehat{\omega}(S)$ on a quantum optical computer is a well studied problem and we limit ourselves here to simply say that it can be done by a sequence of beamsplitters and phase shifters. Appendix \ref{sec:linear_layer_impl_app} provides more details of the procedure and its complexity.

\subsection{Non linearity}
\label{sec:non_linear_layer_impl}

Quantum computers can perform arbitrary computations if given a set of universal gates.
For quantum optical computers, one can take the quadratic Hamiltonians and the cubic Hamiltonian $\widehat{\varphi}^3$ \cite{Lloyd_1999}.
We can implement a non-linearity with Hamiltonian
\eqref{eq:hf} by approximating it with the truncation of the Taylor series of the function $f$ to order $k$:
\begin{align*}
    \widehat{H}^{(k)}
    =
    \sum_{x,a}
    \sum_{\ell=0}^k
    f_\ell 
    \widehat{H}_\ell
    \,,\quad
    \widehat{H}_\ell
    =
    \frac{1}{2}
    \left(
    \widehat{\pi}_{x,a} \widehat{\varphi}_{x,a}^\ell
    +
    \widehat{\varphi}_{x,a}^\ell\widehat{\pi}_{x,a} 
    \right)
    \,.
\end{align*}
We can then use the standard procedures for quantum simulation \cite{NielsenChuang} to implement $\e^{i\widehat{H}^{(k)}}$.
This is detailed in Appendix
\ref{sec:non_linear_layer_impl_app}, where we derive an explicit decomposition in terms of gates which use only the universal Hamiltonians $\widehat{\pi}_{x,a}$, $\widehat{\pi}^2_{x,a}$ and $\widehat{\varphi}_{x,a}^3$. This fully specifies a protocol to implement our model on a quantum optical device.

\section{Conclusion and Outlook}
\label{sec:discussion}

In this paper we introduced a new formalism for deep learning in terms of quantum fields. We made the following contributions: 1) we showed how to use Gaussian states for Bayesian inference, 2) we  devised unitary operators that implement standard non-linearities, 3) we presented tractable limits of the quantum network, 4) we discussed how to implement our models on a quantum computer. 

We implemented the semi-classical architecture to check that it  performed as expected (which it did). 
But because we do not expect that these tractable limits perform any better than classical models we choose not to benchmark them against each other.

Exciting directions for the near future are: studying approximate solutions that get closer to the full quantum model; finding efficient ways to do quantum GP inference on quantum hardware; developing further quantum non-linearities and the quantum formalism for classical models.

%% file: figures/qblob.tex
\begin{tikzpicture}[scale=0.6]
\draw[->] (-1,0) -- (3,0)
node[below right] {$\varphi$};
\draw[->] (0,-1) -- (0,3)
node[left] {$\pi$};
\begin{scope}[rotate=30]
\shadedraw[left color=gray,right color=green, draw=green!50!black]
(2.65,.7) ellipse (1.8cm and 1cm);
\end{scope}
\draw[->] (0,0) -- (2,2)
node[left] {$\bm{m}$};
\end{tikzpicture}

%% file: figures/hierarchy.tex
\begin{tikzpicture}[scale=.7]
\def\angle{60}%
\pgfmathsetlengthmacro{\xoff}{2cm*cos(\angle)}%
\pgfmathsetlengthmacro{\yoff}{1cm*sin(\angle)}%
\draw [thick, fill=gray!10] (\xoff,\yoff) circle[x radius=3cm, y radius=1.5cm] ++(1*\xoff,1*\yoff) node{Quantum};
\draw [thick, fill=gray!50] (0.5*\xoff,0.5*\yoff) circle[x radius=2cm, y radius=1cm] ++(.25*\xoff,.5*\yoff) node{Semiclassical};
\draw [thick, fill=gray!80] (0,0) circle[x radius=1cm, y radius=.5cm] node{Classical};
\draw (2.2*\xoff,.95*\yoff) -- (4.25*\xoff,.95*\yoff);
\draw (6.5*\xoff,.95*\yoff) node{classical nonlinearity};
\draw (1*\xoff,0) -- (4.75*\xoff,0);
\draw (6.5*\xoff,0) node{block diag.~linear};
\end{tikzpicture}

%% file: files/app.tex
\section{Proofs of Section \ref{sec:background}}
\label{sec:proofs_sec_background}

\begin{prop}
\label{prop:C_app}
The covariance matrix $C$ satisfies:
\begin{align}
    C=C^T\,,\quad
    C>0\,,\quad
    C + i J \ge 0\,.    
\end{align}
\end{prop}
\begin{proof}
We give the proof for $|{\cal X}|=1$ and refer the reader to  \cite{PhysRevA.36.3868} for the general case. Suppressing the $x$ index and denoting $\langle \cdot \rangle = \bra{\psi}\cdot \ket{\psi}$, we have:
\begin{align}
C
=
\begin{pmatrix}
2\sigma_{\widehat{\varphi}}^2 & 
\langle \widehat{\varphi}\widehat{\pi}
+\widehat{\pi}\widehat{\varphi}\rangle
-
2
\langle \widehat{\varphi}\rangle
\langle \widehat{\pi} \rangle
\\
\langle \widehat{\varphi}\widehat{\pi}
+\widehat{\pi}\widehat{\varphi}\rangle
-
2
\langle \widehat{\varphi}\rangle
\langle \widehat{\pi} \rangle
&
2\sigma_{\widehat{\pi}}^2
\,,
\end{pmatrix}
\end{align}
where 
$\sigma_{\widehat{\varphi}}^2
=
\langle \widehat{\varphi}^2\rangle
-
\langle \widehat{\varphi}\rangle^2
$ and similarly for $\widehat{\pi}$.
For $|{\cal X}|=1$,
$C + i J\ge 0$ is equivalent to $\Det(C + i J)\ge 0$ or
\begin{align}
&4\sigma_{\widehat{\varphi}}^2    
\sigma_{\widehat{\pi}}^2
\ge
(\langle \widehat{\varphi}\widehat{\pi}
+\widehat{\pi}\widehat{\varphi}\rangle
-
2
\langle \widehat{\varphi}\rangle
\langle \widehat{\pi} \rangle)^2
+1\,.
\end{align}
This is the uncertainty relation in the stronger form due to Robertson–Schroedinger, proving the statement for $|{\cal X}|=1$.
\end{proof}


\begin{prop}
\label{prop:d_omega_trafo_app}
The unitaries of \eqref{eq:D} and \eqref{eq:omega} represent symplectic affine transformations of the canonical operators
\begin{align}
\widehat{D}(\bm{\xi})^\dagger
\widehat{\bm{R}}
\widehat{D}(\bm{\xi})
&=
\widehat{\bm{R}}
+
\bm{\xi}
\\
\widehat{\omega}(S)^\dagger
\widehat{\bm{R}}\,
\widehat{\omega}(S)
&=
S \widehat{\bm{R}}
\,.
\end{align}
\end{prop}
\begin{proof}
The result then follows from the Baker-Campbell-Hausdorff formula 
\begin{align}
\label{eq:bch_app}
\e^A B \e^{-A} = 
\e^{[A,\cdot]}(B) =
B + [A, B] + \tfrac{1}{2}[A,[A,B]] + \cdots
\end{align}
and the commutators:
\begin{align*}
[i\widehat{\bm{R}}^T J\bm{\xi}, \widehat{R}_i]
=
- \xi_i
\,,\quad
[-\frac{i}{2} \widehat{\bm{R}}^T J X 
    \widehat{\bm{R}}, \widehat{R}_i]
=
(X \widehat{\bm{R}})_i
\,.
\end{align*}
\end{proof}

\begin{prop}
\label{prop:gaussian_trafo_app}
Under the unitaries of \eqref{eq:D} and \eqref{eq:omega}, the Gaussian states transform as
\begin{align}
\widehat{D}(\bm{\xi})
\ket{\bm{m}, C}
&=
\ket{\bm{m}+\bm{\xi}, C}
\\
\widehat{\omega}(S)
\ket{\bm{m}, C}
&=
\ket{S\bm{m}, SCS^T}
\,.
\end{align}
\end{prop}
\begin{proof}
From proposition \ref{prop:d_omega_trafo} and the definitions \eqref{eq:mean} and \eqref{eq:cov}, we have for any $\ket{\Psi}$:
\begin{align}
    &\bra{\Psi}\widehat{D}(\bm{\xi})^\dagger \widehat{\bm{R}}\widehat{D}(\bm{\xi}) \ket{\Psi}
    =
    \bra{\Psi} (\widehat{\bm{R}} + \bm{\xi}) \ket{\Psi}
    =
    \bm{m} + \bm{\xi}\nonumber
    \\
    &\bra{\Psi}\widehat{\omega}(S)^\dagger \widehat{\bm{R}}\widehat{\omega}(S) \ket{\Psi}
    =
    S
    \bm{m} =\bm{m}'\nonumber
    \\
    &\bra{\Psi} 
    \widehat{\omega}(S)^\dagger
    \tfrac{1}{2}
    (\widehat{R}_i\widehat{R}_j+\widehat{R}_j\widehat{R}_i)
    \widehat{\omega}(S)\ket{\Psi}
    -
    m'_im'_j \nonumber
    =\\
    & \quad
    \sum_{k,l}
    S_{ik}
    \tfrac{1}{2}C_{kl}    
    S_{jl}
    =
    \tfrac{1}{2}(SCS^T)_{ij}
    \,. \nonumber
\end{align}
The result follows by specifying $\ket{\Psi}=\ket{\bm{m},C}$.
\end{proof}

\section{Proofs of Section \ref{sec:qnn}}
\label{sec:proofs_qnn}

\begin{prop}
\label{prop:ode_phi_app}
The ODEs \eqref{eq:eom_nl_phi}, \eqref{eq:eom_nl_pi} have solutions
\begin{align}
\widehat{\varphi}_{x,a}(t)
&=F^{-1}
\left(F(\widehat{\varphi}_{x,a}(0)) + t\right)
\,,\\
\widehat{\pi}_{x,a}(t)
&=
\frac{1}{2}
\left(
\widehat{\pi}_{x,a}(0)
\frac{f(\widehat{\varphi}_{x,a}(0))}
{f(\widehat{\varphi}_{x,a}(t))}
+\text{h.c.}\right)
\end{align}
where $F'(x) = 1/f(x)$.
\end{prop}
\begin{proof}
We can check directly the formulas by differentiating w.r.t.~$t$ to show that the time evolved fields satisfy the equation of motions.
Rewriting the first equation as
$F(\widehat{\varphi}_{x,a}(t))
=
F(\widehat{\varphi}_{x,a}(0)) + t$ and differentiating the l.h.s.:
\begin{align}
\partial_t
F(\widehat{\varphi}_{x,a}(t))
=
F'(\widehat{\varphi}_{x,a}(t))
\dot{\widehat{\varphi}}_{x,a}(t)
=
\frac{\dot{\widehat{\varphi}}_{x,a}(t)}
{f(\widehat{\varphi}_{x,a}(t))}
\,,
\end{align}
which equals $\partial_t(F(\widehat{\varphi}_{x,a}(0)) + t)=1$, showing that $\widehat{\varphi}_{x,a}(t)$ satisfies  
\eqref{eq:eom_nl_phi}.
For the second equation we differentiate the first term in the parenthesis:
\begin{align}
&\widehat{\pi}_{x,a}(0)
f(\widehat{\varphi}_{x,a}(0))
\partial_t
\left(f(\widehat{\varphi}_{x,a}(t))\right)^{-1}\\
&=
-
\widehat{\pi}_{x,a}(0)
\frac{f(\widehat{\varphi}_{x,a}(0))}
{f(\widehat{\varphi}_{x,a}(t))}
\frac{f'(\widehat{\varphi}_{x,a}(t))}
{f(\widehat{\varphi}_{x,a}(t))}
\dot{\widehat{\varphi}}_{x,a}(t)\\
&=-
\widehat{\pi}_{x,a}(t)
f'(\widehat{\varphi}_{x,a}(t))
\,,
\end{align}
which shows that $\widehat{\pi}_{x,a}(t)$ satisfies \eqref{eq:eom_nl_pi}.
\end{proof}

\section{Proofs of Section \ref{sec:qft_repr}}
\label{sec:proofs_sec_qft_repr}

\begin{lemma}\label{lemma:push_forward_generic_app}
Let $\ket{\bm{m},C}$ be a Gaussian state and $\widehat{U}$ a unitary such that:
\begin{align}
    \label{eq:U_F_app}
    \widehat{U}^\dagger
    \widehat{\bm{\varphi}}
    \widehat{U}
    =
    F(\widehat{\bm{\varphi}})
    \,.
\end{align}
Then:
\begin{align}
    |\bra{\bm{\varphi}}
    \widehat{U}
    \ket{\bm{m}, C}
    |^2
    =
    \left(F \# {\cal GP}(\bm{m}^1, C^{11})\right)(\bm{\varphi})
\end{align}
where $f \# p$ denotes the push forward of $p$ under $f$.
\end{lemma}
\begin{proof}
First we have:
\begin{align}
    |\bra{\bm{\varphi}}
    \ket{\bm{m}, C}
    |^2
    =
    {\cal GP}(\bm{m}^1, C^{11})(\bm{\varphi})
    \,.
\end{align}
Now we note the representation of a projector as the average:
\begin{align}
    \ket{\bm{\varphi}}\bra{\bm{\varphi}}
    =
    \int\frac{D\bm{\lambda}}{(2\pi)^{|{\cal X}|}}
    \e^{i \bm{\lambda}^T (\widehat{\bm{\varphi}} - \bm{\varphi})}
    \,,
\end{align}
a formula which can be proved by comparing matrix elements of the two operators in arbitrary states.
Then we have:
\begin{align}
    &|\bra{\bm{\varphi}}
    \widehat{U}
    \ket{\bm{m}, C}
    |^2
    =
    \bra{\bm{m}, C}
    \widehat{U}^\dagger
    \ket{\bm{\varphi}}\bra{\bm{\varphi}}
    \widehat{U}
    \ket{\bm{m}, C}
    \\
    &=
    \int\frac{D\bm{\lambda}
    \e^{-i \bm{\lambda}^T \bm{\varphi}}
    }{(2\pi)^{|{\cal X}|}}
    \bra{\bm{m}, C}
    \widehat{U}^\dagger 
    \e^{i \bm{\lambda}^T \widehat{\bm{\varphi}}}
    \widehat{U}
    \ket{\bm{m}, C}\\
    &=
    \int\frac{D\bm{\lambda}
    \e^{-i \bm{\lambda}^T \bm{\varphi}}
    }{(2\pi)^{|{\cal X}|}}
    \bra{\bm{m}, C}
    \e^{i \bm{\lambda}^T F(\widehat{\bm{\varphi}})}
    \ket{\bm{m}, C}\\    
    &=
    \int D\bm{\varphi'}
    \int\frac{D\bm{\lambda}
    \e^{i \bm{\lambda}^T (F(\bm{\varphi}')-\bm{\varphi})}
    }{(2\pi)^{|{\cal X}|}}
    |\braket{\bm{\varphi'}}{\bm{m}, C}|^2
    \\
    &=
    \int D\bm{\varphi'}
    \delta(F(\bm{\varphi'})-\bm{\varphi})
    {\cal GP}(\bm{m}^1, C^{11})(\bm{\varphi'})
    \\
    &=
    \left(F \# {\cal GP}(\bm{m}^1, C^{11})\right)(\bm{\varphi})
    \,.
\end{align}
In going from the third to the fourth line we used the following representation, valid for any $f$:
\begin{align}
    f(\widehat{\bm{\varphi}})
    =
    \int D\bm{\varphi}
    f(\bm{\varphi})
    \ket{\bm{\varphi}}\bra{\bm{\varphi}}
    \,,
\end{align}
and we have used the definition of push forward of a distribution, i.e.~ that if $x \sim p_X$, then $p_Z = f\# p_X$ is the distribution of the random variable $z = f(x)$ which can be obtained explicitly via the change of variable formula:
\begin{align}
    p_Z(z)
    &=
    |\Det( \tfrac{\partial f(x)}{\partial x} |_{x = f^{-1}(z)} )
    |^{-1}
    p_X(f^{-1}(z))
    \\
    &=
    \int 
    \dd x'
    \delta(f(x') - z)
    p_X(x')
    \,,
\end{align}
where the second equality follows upon changing variables to $x''=f(x')$.

We also present an alternative proof the lemma, which relies on:
\begin{align}
    \label{eq:ux}
    \widehat{U}\ket{x} = 
    |\Det( \tfrac{\partial F(x)}{\partial x} ) |^{1/2}
    \ket{F(x)}
    \,.
\end{align}
The fact that $\widehat{U}\ket{x} = c(x) \ket{F(x)}$ for some $c(x)$ follows from
\begin{align}
    \widehat{\varphi} \widehat{U}\ket{x} 
    =
    \widehat{U}
    F(\widehat{\varphi}) \ket{x}
    =
    F(x)  \widehat{U}\ket{x}
    \,,
\end{align}
and the delta normalization fixes the proportionality factor:
\begin{align}
\delta(x'-x)
&=
\braket{x'}{x}
=
\bra{x'}\widehat{U}^\dagger \widehat{U}\ket{x}\\
&=
\overline{c(x')} c(x)
\delta(F(x')-F(x))
\,,
\end{align}
using that for a $g(x)$ having the unique solution $g(x)=0$ at $x_0$, one has
\begin{align}
    \delta(g(x))
    =
    \delta(x - x_0)
    |\Det( \tfrac{\partial g(x)}{\partial x}|_{x=x_0} ) |^{-1}
    \,.
\end{align}
Given equation \eqref{eq:ux}, the lemma is immediate:
\begin{align}
    &|\bra{\bm{\varphi}}
    \widehat{U}
    \ket{\bm{m}, C}
    |^2
    =
    |\bra{\bm{m}, C}
    \widehat{U}^\dagger
    \ket{\bm{\varphi}}
    |^2\\
    &=    
    |\Det( \tfrac{\partial F^{-1}(\bm{\varphi}) }{\partial \bm{\varphi}}) |
    \cdot
    |\bra{\bm{m}, C}
    \ket{F^{-1}(\bm{\varphi})}
    |^2\\
    &=
    |\Det( \tfrac{\partial F(x)}{\partial x} |_{x = F^{-1}(\bm{\varphi})} )
    |^{-1}
    {\cal GP}(\bm{m}^1, C^{11})(F^{-1}(\bm{\varphi}))
    \,,
\end{align}
which coincides with the definition of push forward.

\end{proof}

\section{Experiments}
\label{sec:experiments}

In this section we discuss numerical experiments to show that the semi-classical neural networks introduced in Sec.~\ref{sec:classical_simulation} performs on par with its classical counterpart for a simple irregular time series classification task.

\subsection{Datasets}

Similarly to \cite{li2015classification}, to perform controllable experiments with irregular data, we select the UCR time series dataset \cite{UCRArchive}.
We assume that the input time series is uniformly sampled in $[0,1]$ and then sample randomly a fraction $f$ of the input time series, where $f$ varies from $1$ (no subsampling) to $.1$ (only $10\%$ of the data is retained).
This gives us coordinates $x_i$ and values $y_i$ that can be used to prepare an input Gaussian state as discussed below in the training algorithm section.
We select the same random subsampling mask for the the training data and a different one for the test data, which we expect is a good model of a realistic irregular measurement and a more challenging machine learning task.

Since our algorithm has cubic complexity in number of points of the input signal and the purpose of the experiments is only to validate that the semiclassical architecture can perform on par with a classical probabilistic numeric baseline, we restrict the focus here only to the following dataset which has small number of data points:
SyntheticControl (train: $300$, test: $300$, classes: $6$, length time series $60$).

\subsection{Training Algorithm}

The training algorithm is composed of two steps. First, we train the GP using Alg.~\ref{alg:gp}. Then we use the resulting posterior mean and covariance -- more precisely the posterior mean $\mu'$, the cholesky decomposition of the posterior covariance $L$ and its inverse $M$ as inputs to the PNN and SPNN. The PNN discards $M$, which is used for the momenta sector, which is decoupled for PNN and does not intervene in the prediction.
Then we use $\mu', L, M$ to sample from a normal distribution that is defined as in Alg.~\ref{alg:sample}, i.e.~we set all the other means to zero and the covariance is $k'^{-1}=MM^T$ for the momenta sector and channel $0$ while identity elsewhere.
This returns a batch of samples $\bm{z}^i$ that are then passed through the neural network as in Alg.~\ref{alg:spnn} to finally be averaged to produce the logits used for classification.
The next section describes the architectures in more details.

\begin{algorithm}[tb]
   \caption{PretrainGP}
   \label{alg:gp}
\begin{algorithmic}
   \STATE {\bfseries Input:} data $\{x_i^b,y_i^b\}$ $b=1,\dots,B$ ($B$ being data size), grid ${\cal X}$, GP kernel parameters $\Theta$.
   \STATE {\bf Train GP:} 
   Compute $\Theta$ by maximizing marginal likelihood of data in mini-batches:
   \begin{align}
       \max_{\Theta}
       \frac{1}{B}\sum_{b=1}^B \log(p_{GP}(\bm{y}^b|\bm{x}^b, \Theta))
   \end{align}
   \STATE {\bf Test GP:} Compute
   $\mu_x'^b$, $k'^b_{x,x'}$ from Eq.~\eqref{eq:posterior} (main text) for $x,x'\in{\cal X}$.
   \STATE {\bf Cholesky:} Compute $L^b=\text{cholesky}(k'^b)$,
   $M^b=\text{cholesky}((k'^b)^{-1})$
   \STATE {\bf Output:} $\{\mu'^b \}$, $\{ L^b \}$, $\{ M^b \}$
\end{algorithmic}
\end{algorithm}

\begin{algorithm}[tb]
   \caption{Sample}
   \label{alg:sample}
\begin{algorithmic}
   \STATE {\bfseries Input:}$\mu'$, $L$, $M$ (output of Alg.~\ref{alg:gp}), number of channels $N_C$, number of sample $N_{\text{samples}}$.
   \STATE {\bfseries Create Normal:} 
   ${\cal N}(\bm{m}, C)$ with $\bm$: $N_c\times |{\cal X}|\times 2$ and 
   $C$: $(N_c\times |{\cal X}|\times 2) \times (N_c\times |{\cal X}|\times 2)$,
   \begin{align*}
       &m_{x,c,j}
       =
       \begin{cases}
       \mu'_x & c=j=0\\
       0 & \text{else}
       \end{cases}\\
       &C_{x,c,j;x',c',j'}
       =
       \begin{cases}
        k'_{x,x'} & j=j'=c=c'=0 \\
        (k'^{-1})_{x,x'} & j=j'=1, c\!=\!c'\!\!=\!0\\
        \delta_{x,x'}\delta_{c,c'}\delta_{j,j'}
        & \text{else}
       \end{cases}
   \end{align*}
   with $k'=LL^T$, $k'^{-1}=MM^T$:
   \STATE {\bfseries Sample:} 
   $\bm{z}^i$ from ${\cal N}(\bm{m}, C)$ using the reparametrization trick for $i=1,\dots,N_{\text{samples}}$
   \STATE {\bfseries Output:} $\{\bm{z}^i\}$
\end{algorithmic}
\end{algorithm}

\begin{algorithm}[tb]
   \caption{Semi-ClassicalNN}
   \label{alg:spnn}
\begin{algorithmic}
   \STATE {\bfseries Input:}$\mu'$, $L$, $M$ (output of Alg.~\ref{alg:gp}), number of channels $N_C$, $\beta$, number of sample $N_{\text{samples}}$, learnable parameters $h^{(\ell)}, b^{(\ell)}$
   \STATE {\bfseries Sample:} 
   $\bm{z}^i = Sample(\mu', L, M, N_C, N_{\text{samples}})$ (see Alg.~\ref{alg:sample}).
   \FOR{$\ell=0$ {\bfseries to} $L-1$}
   \STATE $\bm{z}^i = \text{Linear}(\bm{z}^i, h^{(\ell)}, b^{(\ell)})$
   \STATE $\bm{z}^i = \text{SympecticSofplus}(\bm{z}^i, \beta)$
   \ENDFOR
   \STATE $\bm{z}^i = \text{Linear}(\bm{z}^i, h^{(L)}, b^{(L)})$
   \STATE $l_c = N_{\text{samples}}^{-1}
   \sum_{i=1}^{N_{\text{samples}}}
   z^i_{x=0,c,j=0}$
   \STATE {\bfseries Output:} Logits $l_c$ 
\end{algorithmic}
\end{algorithm}

Note that GP inference is the bottleneck here and we could use the methods of \cite{li2016scalable} to improve on that. Our interest is however not to have a state of the art classifier, so we test the model in the simplest setting of $O(N^3)$ exact GP inference and simply validate that the performance of the semi-classical neural network architecture is on par with a classical counterpart.
Note that separating the GP training from the neural network training allows us to amortize the GP inference as a preprocessing step and have an $O(N^2)$ algorithm for network propagation as in BNN case.

\subsection{Architecture}

For the GP to interpolate the data we take Matern kernel with $\nu=0.5$ as implemented in gpytorch \cite{gpytorch}.
We compare three models: 1) a baseline MLP (BNN), 2) a probabilistic numeric MLP (PNN) and 3) a semiclassical version of the PNN (SPNN).
All models share the same basic architecture, which is a MLP with three layers. In all cases the hidden sizes are equal to the number of points in the original grid of the data, $|{\cal X}|$, times the number of channels $N_C$.
For BNN we parameterize the weight matrices directly i.e.~the learnable parameters are the weight matrices. Instead for PNN and SPNN we parametrize the logarithm of the weight matrix as the learnable parameters. We did check that for PNN this did not impact perform comparing to the case where we use the same parameterization of BNN.
We do this to ensure simply that the weight matrices of the SPNN are symplectic. Indeed recalling the definitions from the main text, see Sec.~\ref{sec:quantum}, a linear layer with input/output dimensions $2N$ has learnable parameters $A, B, C\in \mathbb{R}^{N\times N}$ and the weight matrix $S$ is then constructed as:
\begin{align}
    S = \exp
    \begin{pmatrix}
    A & \tfrac{1}{2}(B+B^T) \\
    \tfrac{1}{2}(C+C^T) & -A^T
    \end{pmatrix}
    \,.
\end{align}
$A, B, C$ are denoted $h$ in Alg.~\ref{alg:spnn}.
For PNN we simply take $B=C=0$ and discard the lower right block.
We have also tried to restrict $S$ to be a free matrix and use the parametrization from \cite{de2006symplectic} but found that the matrix exponential (which allow use to parametrize more general symplectic matrices) worked best.
We used pytorch \cite{pytorch} to run the experiments and initialized all the matrices $W$ using \verb| nn.init.kaiming_uniform_|, with mode fanin and nonlinearity 'relu'.
For SPNN we further multiplied the learnable parameters by a scale factor $0.1$ which we found important with the training settings described below to prevent large values at the beginning of training.

Finally, we took for non-linearity the softplus for BNN and PNN and its symplectic version of section \ref{lemma:s_softplus} for SPNN.
In all cases, $\beta=0.1$ was chosen. 
This ensures that in the SPNN we do not create very large values for large negative $\varphi$ that would make training unstable, and was chosen the same for all models for comparison. We checked that one could get similar results for BNN and PNN by using the more conventional value of $\beta=1$.

\subsection{Results}

We present results in table \ref{tab:results}.
We pretrained the GPs using Adam optimizer with learning rate $0.1$ from \cite{pytorch} for $20$ epochs using batch size of $50$.
We trained the neural network using default Adam optimizer for $1000$ epochs with batch size $50$ as well.
We used $N_{\text{samples}} = 100$ and $N_C=2$.

The BNN baseline as well as PNN and SPNN perform on the original unsampled dataset similarly to the MLP used in \cite{wang2016time} ($95$ for SyntheticControl), which also has three layers but additionally dropout and uses ReLU instead of softplus. This validates the choice of architecture we made.

The results presented are on the average over three random seeds for the random sampling. The error bars are considerable in all cases due to the fact that different sampling masks lead to different data points that are retained and this choice affects the classification of the signal.
When we start to subsample the data (using the same procedure for all models), both the PNN and SPNN n average perform better than the BNN, confirming the findings of \cite{li2015classification,li2016scalable,finzi2020probabilistic} showing the the usefulness of the GP model for irregularly sampled data. 
\cite{li2015classification} report accuracies averaged over all the 43 UCR datasets and not directly comparable with our restricted setting, while \cite{li2016scalable,finzi2020probabilistic} use different datasets for which scalable methods for GP inference are required.
Finally, we note that the performance of PNN and SPNN have a considerable overlap within the confidence region of the results, validating the claim that they perform similarly.

\begin{table}[t]
\caption{Classification accuracies for a baseline MLP (BNN), the classical probabilistic numeric network (PNN) and the semiclassical  network (SPNN). Sampling stands for the subsampling fraction $f$ of the input as explained in the main text. The results are obtained by computing statistics over three random subsampling masks.
}
\label{tab:results}
\vskip 0.15in
\begin{center}
\begin{small}
\begin{sc}
\begin{tabular}{cccc}
\toprule
Sampling & BNN & PNN & SPNN  \\
\midrule
1 & 94 & 94 & 95 \\
0.9 & 84.66 $\pm$ 0.47 & 93.33 $\pm$ 1.69 & 93.00 $\pm$ 2.82 \\
0.8 & 80.66 $\pm$ 3.68 & 91.00 $\pm$ 2.82 & 91.33 $\pm$ 2.05\\
0.7 & 77.33 $\pm$ 6.02 & 89.00 $\pm$ 2.16 & 86.66 $\pm$ 1.25\\
0.6 & 74.99 $\pm$ 7.87 & 82.33 $\pm$ 5.31 & 82.00 $\pm$ 2.94\\
0.5 & 68.33 $\pm$ 6.85 & 75.99 $\pm$ 8.49 & 74.33 $\pm$ 11.12\\
0.4 & 68.33 $\pm$ 5.44 & 72.66 $\pm$ 3.68 & 72.33 $\pm$ 4.11\\
0.3 & 66.33 $\pm$ 5.25 & 66.99 $\pm$ 4.55 & 66.66 $\pm$ 5.44\\
0.2 & 38.99 $\pm$ 8.04 & 42.33 $\pm$ 5.25 & 41.66 $\pm$ 4.78\\
0.1 & 28.33 $\pm$ 8.18 & 36.00 $\pm$ 5.35 & 40.33 $\pm$ 8.05\\
\bottomrule
\end{tabular}
\end{sc}
\end{small}
\end{center}
\vskip -0.1in
\end{table}

\section{Details of quantum optical implementation}

\subsection{State preparation strategies}
\label{sec:state_prep_app}

The state preparation step encodes the input signal ${\cal D}=\{(x_i, y_i)\}_{i=1}^N$ into quantum registers. 
There a several options for that. 
First we note that while the prescription of lemma \ref{lemma:qgp}, i.e.~to create a prior state and to project onto $y_i$ the registers associated to locations $x_i$, is appealing conceptually and allows us to formulate inference entirely in quantum language, it is not straightforward to implement. This is because projections are enacted by partial measurements and those give outcome $y_i$ only with a small probability of success. To overcome this, we consider the following alternatives.
The simplest is the one explained in the main text, that we repeat here. We use classical hardware to compute the GP posterior as in \eqref{eq:posterior} for a set of points $x\in {\cal X}$. Then we create a an input Gaussian state by acting on the vacuum state $\ket{\Omega}$ defined in Sec.~\ref{sec:hintons} with the linear layer $\widehat{D}(\bm{\xi}=(\bm{\mu}', 0))\widehat{\omega}(S=A\oplus (A)^{-1})$, where $A$ is a square root of $k'$.
The vacuum state can be created using lasers \cite{NielsenChuang} and we defer to the next section the implementation of the linear layer. 
While this procedure is straightforward it incurs a complexity similar to the classical GP inference, that is $O(N^3)$.
Another alternative is to use the fact that on quantum hardware we can perform matrix inversion exponentially faster \cite{Harrow_2009}. Under some conditions (in particular access to QRAM) one can compute posterior mean and covariance in polylog time \cite{zhao2019quantum, Das_2018}. Under similar conditions, quantum singular value decomposition can compute $A$ in polylog time \cite{kerenidis2016quantum}.
The catch in our setting however is that to read out the values of $\bm{\mu}'$ and $A$ we would still need $O(N^2)$ operations.
We leave developing a more efficient quantum implementation of quantum GP inference as in interesting future challenge.

\subsection{Linear layer}
\label{sec:linear_layer_impl_app}

The implementation of a unitary $\widehat{\omega}(S)$ on a quantum optical computer is a well studied problem and we limit here to give a high level description, see \cite{Weedbrook_2012} for more details. We first decompose the unitary in terms of elementary linear optical gates, we use the group homomorphism property $\widehat{\omega}(S)\widehat{\omega}(S')=\widehat{\omega}(SS')$ together with the Bloch-Messiah decomposition  
$S = K \Sigma L$ with $K,L$ symplectic and orthogonal and $\Sigma=\text{diag}(\e^{r_1},\dots,\e^{r_M},\e^{-r_1},\dots,\e^{-r_M})$. 
$\Sigma$ can be implemented directly using optical parametric amplifiers.
The orthogonal matrices $K,L$ can be further decomposed using Givens rotations as product of rotations of two components and implemented in terms of beamsplitters and phase shifters.
Note that the number of gates required to decompose an arbitrary matrix can grow quadratically with the dimension \cite{PhysRevLett.73.58}.
The bias $\widehat{D}$ can be also easily implemented \cite{Weedbrook_2012}.

\subsection{Non linearity}
\label{sec:non_linear_layer_impl_app}

Quantum computers can perform arbitrary computations if given a set of universal gates.
For quantum optical computers, one can take the quadratic Hamiltonians and the cubic gate, whose Hamiltonian is $\widehat{\varphi}^3$ \cite{Lloyd_1999}.
We can implement a non-linearity with Hamiltonian
\eqref{eq:hf} by approximating it with the truncation of the Taylor series of the function $f$ to order $k$:
\begin{align}
    \widehat{H}
    =
    \sum_{\ell=0}^k
    f_\ell 
    \widehat{H}_\ell
    \,,\quad
    \widehat{H}_\ell
    =
    \frac{1}{2}
    \left(
    \widehat{\pi} \widehat{\varphi}^\ell
    +
    \widehat{\varphi}^\ell\widehat{\pi} 
    \right)
    \,.
\end{align}
where for notational simplicity in this section we are going to omit the indices $x,a$ from the quantum fields.
We can then use the standard procedure for quantum simulation, see e.g.~\cite{NielsenChuang}.
A first step is to trade $\e^{i m \epsilon \widehat{H}}$ for the $m$ times application of $\prod_{\ell=0}^k \e^{i \epsilon f_\ell \widehat{H}_\ell}$, so that the problem boils down on how to implement $\e^{i \epsilon f_\ell \widehat{H}_\ell}$.
This is explained in the next proposition:
\begin{prop}
Denoted $\alpha_\ell = -1/(4(\ell+2))$, define the unitaries
\begin{align}
&\widehat{U}_{\ell,\epsilon} 
= \exp(i \epsilon \widehat{H}_\ell) 
\\   
&\widehat{W}_{\ell,\epsilon} 
=
\exp(i \epsilon \alpha_\ell \widehat{\pi}^2)
\,,\\
&\widehat{V}_{\ell,\epsilon} 
=
\exp(-i \epsilon \widehat{\varphi}^3)
\widehat{U}_{\ell,\epsilon}^\dagger
\exp(i \epsilon \widehat{\varphi}^3)
\widehat{U}_{\ell,\epsilon}
\,.
\end{align}
We have
\begin{align}
\widehat{U}_{\ell+1, \epsilon}
&=
\widehat{W}_{\ell,\sqrt{\epsilon}} 
\widehat{V}_{\ell, \epsilon^{1/4}}
\widehat{W}_{\ell,\sqrt{\epsilon}} ^\dagger
\widehat{V}_{\ell, \epsilon^{1/4}}^\dagger
+ O(\epsilon^{3/4})
\,.
\end{align}
\end{prop}
\begin{proof}
Consider the Trotter formula:
\begin{align*}
    &
    [\e^{i\epsilon \widehat{H}_1}\e^{i\epsilon \widehat{H}_2} \cdots \e^{i\epsilon \widehat{H}_k}]
    [\e^{i\epsilon \widehat{H}_k}\e^{i\epsilon \widehat{H}_{k-1}} \cdots \e^{i\epsilon \widehat{H}_1}]\\
    &=
    \e^{2 i \epsilon \widehat{H}}
    +
    O(\epsilon^3)
    \,.
\end{align*}
The error in replacing $\e^{i \widehat{H}}$ with the $m$-th power of the l.h.s.~is smaller than $\alpha m \epsilon^3$ for some constant $\alpha$ (see Eq.~4.107 of \citep{NielsenChuang}).
So we can concentrate on implementing $\e^{i\epsilon H_\ell}$. Suppose that we know how to implement $\e^{i\epsilon H_{\ell}}$. Then we show how to construct $\e^{i\epsilon H_{\ell+1}}$ using the universal cubic and linear gates. Note that
\begin{align*}
    &[\widehat{\varphi}^3, \widehat{H_{\ell}} ]
    =
    f_\ell
    [\widehat{\varphi}^3,
    \widehat{\pi} \widehat{\varphi}^\ell 
    ]
    =
    f_\ell \widehat{\varphi}^\ell
    [\widehat{\varphi}^3,
    \widehat{\pi} 
    ]
    =
    2i f_\ell
    \widehat{\varphi}^{\ell}
    \,,\\
    &[\widehat{\pi}^2, [\widehat{\varphi}^3, \widehat{H}_{\ell} ]]
    =
    2if_\ell
    [\widehat{\pi}^2, \widehat{\varphi}^{\ell+2}]
    \\
    &=
    2
    f_\ell(\ell+2)
    (\widehat{\pi}
    \widehat{\varphi}^{\ell+1}
    +
    \widehat{\varphi}^{\ell+1}
    \widehat{\pi})\\
    &=
    4\frac{f_\ell}{f_{\ell+1}}(\ell+2)
    \widehat{H}_{\ell+1}
    \,.
\end{align*}
Now if we know how to implement two gate $\e^{i\epsilon A}, \e^{i\epsilon B}$, we also know how to implement the gate with Hamiltonian given by their commutator:
\begin{align*}
    &\e^{-i\sqrt{\epsilon} \widehat{A}}\e^{-i\sqrt{\epsilon} \widehat{B}}
    \e^{+i\sqrt{\epsilon} \widehat{A}}\e^{+i\sqrt{\epsilon} \widehat{B}}\\
    &=
    \e^{-i\sqrt{\epsilon} (\widehat{A}+\widehat{B}) - \frac{1}{2}\epsilon [\widehat{A},\widehat{B}]}
    \e^{+i\sqrt{\epsilon} (\widehat{A}+\widehat{B}) - \frac{1}{2}\epsilon [\widehat{A},\widehat{B}]}
    +
    O(\epsilon^{3/2})\\
    &=
    \e^{-\epsilon [\widehat{A},\widehat{B}]}
    +
    O(\epsilon^{3/2})
    \,.
\end{align*}
So if we know how to implement $\e^{i\epsilon H_\ell}$, we can implement $\e^{i\epsilon H_{\ell+1}}$ using the fact that $\widehat{\pi}^2$ and $\widehat{\varphi}^3$ are universal gates and using the above formula two times.
Since we know how to implement $H_{\ell=0}=\widehat{\pi}$, we can implement all the higher Hamiltonians recursively. 
\end{proof}

The number $N_{\widehat{H}}$ of universal gates to implement a gate with Hamiltonian $\widehat{H}$ satisfies the recursion:
\begin{align}
N_{\widehat{H}_{\ell+1}} = 2N_{\widehat{\pi}^2} + 2(2N_{\widehat{\varphi}^3} + 2 N_{\widehat{H}_{\ell}})
\,,
\end{align}
where setting $N_{\widehat{\pi}^2}=N_{\widehat{\varphi}^2}=N_{\widehat{H}_{0}}=1$ we get the recursion $N_{\widehat{H}_{\ell+1}} = 6 + 4 N_{\widehat{H}_{\ell}}, N_{\widehat{H}_{0}}=1$ whose solution is exponential in $\ell$, and therefore in $k$, the truncation parameter of the non-linearity. 
Depending on the hardware, a low value of $k$ might be required with the proposed procedure.

To understand what function the approximated gates implement, let us consider:
\begin{align*}
    &h_{\ell,\alpha}(\widehat{\varphi})
    :=
    \e^{i \frac{\alpha}{2} (\widehat{\pi} \widehat{\varphi}^\ell + \widehat{\varphi}^\ell \widehat{\pi})}
    \widehat{\varphi}
    \e^{-i \frac{\alpha}{2} (\widehat{\pi} \widehat{\varphi}^\ell + \widehat{\varphi}^\ell \widehat{\pi})}\\
    &=
    \widehat{\varphi}
    +
    \alpha \widehat{\varphi}^\ell
    +
    \frac{\alpha^2}{2}
    \ell \widehat{\varphi}^{\ell-1} \widehat{\varphi}^{\ell}
    +
    \frac{\alpha^3}{3!}
    \ell(2\ell-1) \widehat{\varphi}^{2\ell-2} \widehat{\varphi}^{\ell}
    +\dots
    \\
    &=
    \widehat{\varphi}
    \sum_{j\ge 0}
    \frac{\alpha^j}{j!}
    \ell(2\ell-1)(3\ell-2) \dots ((j-1)\ell - (j-2))
    \widehat{\varphi}^{j (\ell -1)}
    \,,\\
    &h_{1,\alpha}(\widehat{\varphi}) = \e^{\alpha} \widehat{\varphi}\\
    &h_{2,\alpha}(\widehat{\varphi}) = \frac{\widehat{\varphi}}{1-\alpha \widehat{\varphi}}\\
    &h_{3,\alpha}(\widehat{\varphi}) = \frac{\widehat{\varphi}}{\sqrt{1-2 \alpha \widehat{\varphi}^2}}\\
    &\dots
    \,.
\end{align*}
So consider the Hamiltonian of softplus by setting
\begin{align}
    \alpha_\ell = \epsilon \frac{(-1)^\ell}{\ell!}
\end{align}
and truncating to order $k$ and using Trotter, the function implemented by our procedure is:
\begin{align}
&[\e^{i\epsilon H_1}\cdots \e^{i\epsilon H_k}]
[\e^{i\epsilon H_k} \cdots \e^{i\epsilon H_1}]\cdot\\
&\widehat{\varphi}\cdot
[\e^{-i\epsilon H_1}\cdots \e^{-i\epsilon H_k}]
[\e^{-i\epsilon H_k} \cdots \e^{-i\epsilon H_1}]\\
&=
h_{1,\alpha_1} \circ \cdots \circ
h_{k,\alpha_k} \circ h_{k,\alpha_k} \circ 
\cdots \circ h_{1,\alpha_1}
(\widehat{\varphi})
\,.
\end{align}
Define the $m$-th power of this function by $\sigma_{k,\epsilon, m}(\widehat{\varphi})$.
Figure \ref{fig:sigma_trunc} shows these non-linearity against the original softplus function for different values of the parameters showing the effect of the choice of $m$ and $\epsilon$ for $k=3$.

\begin{figure}[ht]
\begin{center}
\centerline{\includegraphics[width=\columnwidth]{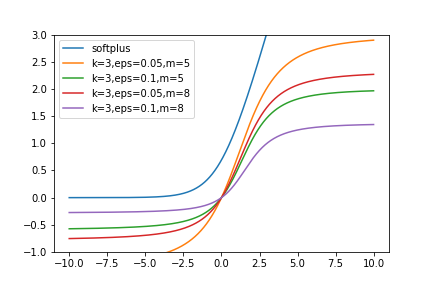}}
\caption{Truncated non-linearity $\sigma_{k,\epsilon,m}(x)$}
\label{fig:sigma_trunc}
\end{center}
\vskip -0.2in
\end{figure}

%% file: icml2021.bbl
\begin{thebibliography}{45}
\providecommand{\natexlab}[1]{#1}
\providecommand{\url}[1]{\texttt{#1}}
\expandafter\ifx\csname urlstyle\endcsname\relax
  \providecommand{\doi}[1]{doi: #1}\else
  \providecommand{\doi}{doi: \begingroup \urlstyle{rm}\Url}\fi

\bibitem[Adesso et~al.(2014)Adesso, Ragy, and Lee]{adesso2014continuous}
Adesso, G., Ragy, S., and Lee, A.~R.
\newblock Continuous variable quantum information: Gaussian states and beyond.
\newblock \emph{Open Systems \& Information Dynamics}, 21\penalty0
  (01n02):\penalty0 1440001, 2014.

\bibitem[Bartlett et~al.(2002)Bartlett, Sanders, Braunstein, and
  Nemoto]{Bartlett_2002}
Bartlett, S.~D., Sanders, B.~C., Braunstein, S.~L., and Nemoto, K.
\newblock Efficient classical simulation of continuous variable quantum
  information processes.
\newblock \emph{Physical Review Letters}, 88\penalty0 (9), Feb 2002.
\newblock ISSN 1079-7114.
\newblock \doi{10.1103/physrevlett.88.097904}.
\newblock URL \url{http://dx.doi.org/10.1103/PhysRevLett.88.097904}.

\bibitem[Bartlett et~al.(2012)Bartlett, Rudolph, and Spekkens]{Bartlett_2012}
Bartlett, S.~D., Rudolph, T., and Spekkens, R.~W.
\newblock Reconstruction of gaussian quantum mechanics from liouville mechanics
  with an epistemic restriction.
\newblock \emph{Physical Review A}, 86\penalty0 (1), Jul 2012.
\newblock ISSN 1094-1622.
\newblock \doi{10.1103/physreva.86.012103}.
\newblock URL \url{http://dx.doi.org/10.1103/PhysRevA.86.012103}.

\bibitem[{Beer} et~al.(2019){Beer}, {Bondarenko}, {Farrelly}, {Osborne},
  {Salzmann}, and {Wolf}]{beer}
{Beer}, K., {Bondarenko}, D., {Farrelly}, T., {Osborne}, T.~J., {Salzmann}, R.,
  and {Wolf}, R.
\newblock {Efficient Learning for Deep Quantum Neural Networks}.
\newblock \emph{arXiv e-prints}, art. arXiv:1902.10445, February 2019.

\bibitem[Bondesan \& Lamacraft(2019)Bondesan and
  Lamacraft]{bondesan2019learning}
Bondesan, R. and Lamacraft, A.
\newblock Learning symmetries of classical integrable systems, 2019.

\bibitem[Bondesan \& Welling(2020)Bondesan and Welling]{bondesan2020quantum}
Bondesan, R. and Welling, M.
\newblock Quantum deformed neural networks.
\newblock 2020.

\bibitem[Chen et~al.(2015)Chen, Keogh, Hu, Begum, Bagnall, Mueen, and
  Batista]{UCRArchive}
Chen, Y., Keogh, E., Hu, B., Begum, N., Bagnall, A., Mueen, A., and Batista, G.
\newblock The ucr time series classification archive, July 2015.
\newblock \url{www.cs.ucr.edu/~eamonn/time_series_data/}.

\bibitem[Cockayne et~al.(2019)Cockayne, Oates, Sullivan, and
  Girolami]{cockayne2019bayesian}
Cockayne, J., Oates, C.~J., Sullivan, T.~J., and Girolami, M.
\newblock Bayesian probabilistic numerical methods.
\newblock \emph{SIAM Review}, 61\penalty0 (4):\penalty0 756--789, 2019.

\bibitem[Cohen et~al.(2018)Cohen, Geiger, Koehler, and
  Welling]{cohen2018spherical}
Cohen, T.~S., Geiger, M., Koehler, J., and Welling, M.
\newblock Spherical cnns, 2018.

\bibitem[{Cong} et~al.(2019){Cong}, {Choi}, and {Lukin}]{qcnn}
{Cong}, I., {Choi}, S., and {Lukin}, M.~D.
\newblock {Quantum convolutional neural networks}.
\newblock \emph{Nature Physics}, 15\penalty0 (12):\penalty0 1273--1278, August
  2019.
\newblock \doi{10.1038/s41567-019-0648-8}.

\bibitem[Das et~al.(2018{\natexlab{a}})Das, Siopsis, and Weedbrook]{Das_2018}
Das, S., Siopsis, G., and Weedbrook, C.
\newblock Continuous-variable quantum gaussian process regression and quantum
  singular value decomposition of nonsparse low-rank matrices.
\newblock \emph{Physical Review A}, 97\penalty0 (2), Feb 2018{\natexlab{a}}.
\newblock ISSN 2469-9934.
\newblock \doi{10.1103/physreva.97.022315}.
\newblock URL \url{http://dx.doi.org/10.1103/PhysRevA.97.022315}.

\bibitem[Das et~al.(2018{\natexlab{b}})Das, Siopsis, and
  Weedbrook]{PhysRevA.97.022315}
Das, S., Siopsis, G., and Weedbrook, C.
\newblock Continuous-variable quantum gaussian process regression and quantum
  singular value decomposition of nonsparse low-rank matrices.
\newblock \emph{Phys. Rev. A}, 97:\penalty0 022315, Feb 2018{\natexlab{b}}.
\newblock \doi{10.1103/PhysRevA.97.022315}.
\newblock URL \url{https://link.aps.org/doi/10.1103/PhysRevA.97.022315}.

\bibitem[de~Gosson(2006)]{de2006symplectic}
de~Gosson, M.
\newblock \emph{Symplectic Geometry and Quantum Mechanics}.
\newblock Operator Theory: Advances and Applications. Birkh{\"a}user Basel,
  2006.
\newblock ISBN 9783764375751.
\newblock URL \url{https://books.google.nl/books?id=q9SHRvay75IC}.

\bibitem[De~Gosson(2009)]{de2009symplectic}
De~Gosson, M.~A.
\newblock The symplectic camel and the uncertainty principle: The tip of an
  iceberg?
\newblock \emph{Foundations of Physics}, 39\penalty0 (2):\penalty0 194--214,
  2009.

\bibitem[Dupont et~al.(2019)Dupont, Doucet, and Teh]{dupont2019augmented}
Dupont, E., Doucet, A., and Teh, Y.~W.
\newblock Augmented neural odes.
\newblock 2019.

\bibitem[{Farhi} \& {Neven}(2018){Farhi} and {Neven}]{farhi_neven}
{Farhi}, E. and {Neven}, H.
\newblock {Classification with Quantum Neural Networks on Near Term
  Processors}.
\newblock \emph{arXiv e-prints}, art. arXiv:1802.06002, February 2018.

\bibitem[Finzi et~al.(2020)Finzi, Bondesan, and
  Welling]{finzi2020probabilistic}
Finzi, M., Bondesan, R., and Welling, M.
\newblock Probabilistic numeric convolutional neural networks.
\newblock 2020.

\bibitem[Gardner et~al.(2018)Gardner, Pleiss, Bindel, Weinberger, and
  Wilson]{gpytorch}
Gardner, J.~R., Pleiss, G., Bindel, D., Weinberger, K.~Q., and Wilson, A.~G.
\newblock Gpytorch: Blackbox matrix-matrix gaussian process inference with gpu
  acceleration.
\newblock In \emph{Advances in Neural Information Processing Systems}, 2018.

\bibitem[Harrow et~al.(2009)Harrow, Hassidim, and Lloyd]{Harrow_2009}
Harrow, A.~W., Hassidim, A., and Lloyd, S.
\newblock Quantum algorithm for linear systems of equations.
\newblock \emph{Physical Review Letters}, 103\penalty0 (15), Oct 2009.
\newblock ISSN 1079-7114.
\newblock \doi{10.1103/physrevlett.103.150502}.
\newblock URL \url{http://dx.doi.org/10.1103/PhysRevLett.103.150502}.

\bibitem[Hudson(1974)]{HUDSON1974249}
Hudson, R.
\newblock When is the wigner quasi-probability density non-negative?
\newblock \emph{Reports on Mathematical Physics}, 6\penalty0 (2):\penalty0 249
  -- 252, 1974.
\newblock ISSN 0034-4877.
\newblock \doi{https://doi.org/10.1016/0034-4877(74)90007-X}.
\newblock URL
  \url{http://www.sciencedirect.com/science/article/pii/003448777490007X}.

\bibitem[{Huggins} et~al.(2019){Huggins}, {Patil}, {Mitchell}, {Whaley}, and
  {Miles Stoudenmire}]{huggins}
{Huggins}, W., {Patil}, P., {Mitchell}, B., {Whaley}, K.~B., and {Miles
  Stoudenmire}, E.
\newblock {Towards quantum machine learning with tensor networks}.
\newblock \emph{Quantum Science and Technology}, 4\penalty0 (2):\penalty0
  024001, Apr 2019.
\newblock \doi{10.1088/2058-9565/aaea94}.

\bibitem[Kerenidis \& Prakash(2016)Kerenidis and Prakash]{kerenidis2016quantum}
Kerenidis, I. and Prakash, A.
\newblock Quantum recommendation systems, 2016.

\bibitem[Killoran et~al.(2019)Killoran, Bromley, Arrazola, Schuld, Quesada, and
  Lloyd]{killoran2019}
Killoran, N., Bromley, T.~R., Arrazola, J.~M., Schuld, M., Quesada, N., and
  Lloyd, S.
\newblock Continuous-variable quantum neural networks.
\newblock \emph{Physical Review Research}, 1\penalty0 (3):\penalty0 033063,
  2019.

\bibitem[Lau et~al.(2017)Lau, Pooser, Siopsis, and
  Weedbrook]{PhysRevLett.118.080501}
Lau, H.-K., Pooser, R., Siopsis, G., and Weedbrook, C.
\newblock Quantum machine learning over infinite dimensions.
\newblock \emph{Phys. Rev. Lett.}, 118:\penalty0 080501, Feb 2017.
\newblock \doi{10.1103/PhysRevLett.118.080501}.
\newblock URL \url{https://link.aps.org/doi/10.1103/PhysRevLett.118.080501}.

\bibitem[Li \& Marlin(2016)Li and Marlin]{li2016scalable}
Li, S. C.-X. and Marlin, B.
\newblock A scalable end-to-end gaussian process adapter for irregularly
  sampled time series classification, 2016.

\bibitem[Li \& Marlin(2015)Li and Marlin]{li2015classification}
Li, S. C.-X. and Marlin, B.~M.
\newblock Classification of sparse and irregularly sampled time series with
  mixtures of expected gaussian kernels and random features.
\newblock In \emph{UAI}, pp.\  484--493, 2015.

\bibitem[Lloyd \& Braunstein(1999)Lloyd and Braunstein]{Lloyd_1999}
Lloyd, S. and Braunstein, S.~L.
\newblock Quantum computation over continuous variables.
\newblock \emph{Physical Review Letters}, 82\penalty0 (8):\penalty0
  1784–1787, Feb 1999.
\newblock ISSN 1079-7114.
\newblock \doi{10.1103/physrevlett.82.1784}.
\newblock URL \url{http://dx.doi.org/10.1103/PhysRevLett.82.1784}.

\bibitem[Maron et~al.(2018)Maron, Ben-Hamu, Shamir, and
  Lipman]{maron2018invariant}
Maron, H., Ben-Hamu, H., Shamir, N., and Lipman, Y.
\newblock Invariant and equivariant graph networks.
\newblock \emph{arXiv preprint arXiv:1812.09902}, 2018.

\bibitem[Massaroli et~al.(2021)Massaroli, Poli, Park, Yamashita, and
  Asama]{massaroli2021dissecting}
Massaroli, S., Poli, M., Park, J., Yamashita, A., and Asama, H.
\newblock Dissecting neural odes.
\newblock 2021.

\bibitem[Nielsen \& Chuang(2000)Nielsen and Chuang]{NielsenChuang}
Nielsen, M. and Chuang, I.
\newblock \emph{Quantum Computation and Quantum Information}.
\newblock Cambridge Series on Information and the Natural Sciences. Cambridge
  University Press, 2000.
\newblock ISBN 9780521635035.
\newblock URL \url{https://books.google.co.uk/books?id=aai-P4V9GJ8C}.

\bibitem[Paszke et~al.(2019)Paszke, Gross, Massa, Lerer, Bradbury, Chanan,
  Killeen, Lin, Gimelshein, Antiga, Desmaison, Kopf, Yang, DeVito, Raison,
  Tejani, Chilamkurthy, Steiner, Fang, Bai, and Chintala]{pytorch}
Paszke, A., Gross, S., Massa, F., Lerer, A., Bradbury, J., Chanan, G., Killeen,
  T., Lin, Z., Gimelshein, N., Antiga, L., Desmaison, A., Kopf, A., Yang, E.,
  DeVito, Z., Raison, M., Tejani, A., Chilamkurthy, S., Steiner, B., Fang, L.,
  Bai, J., and Chintala, S.
\newblock Pytorch: An imperative style, high-performance deep learning library.
\newblock In Wallach, H., Larochelle, H., Beygelzimer, A., d\textquotesingle
  Alch\'{e}-Buc, F., Fox, E., and Garnett, R. (eds.), \emph{Advances in Neural
  Information Processing Systems 32}, pp.\  8024--8035. Curran Associates,
  Inc., 2019.
\newblock URL
  \url{http://papers.neurips.cc/paper/9015-pytorch-an-imperative-style-high-performance-deep-learning-library.pdf}.

\bibitem[Pinkus(1999)]{pinkus_1999}
Pinkus, A.
\newblock Approximation theory of the mlp model in neural networks.
\newblock \emph{Acta Numerica}, 8:\penalty0 143–195, 1999.
\newblock \doi{10.1017/S0962492900002919}.

\bibitem[Rasmussen et~al.(2006)Rasmussen, Williams, Press, Bach, and
  (Firm)]{rasmussen2006gaussian}
Rasmussen, C., Williams, C., Press, M., Bach, F., and (Firm), P.
\newblock \emph{Gaussian Processes for Machine Learning}.
\newblock Adaptive computation and machine learning. MIT Press, 2006.
\newblock ISBN 9780262182539.
\newblock URL \url{https://books.google.nl/books?id=Tr34DwAAQBAJ}.

\bibitem[Reck et~al.(1994)Reck, Zeilinger, Bernstein, and
  Bertani]{PhysRevLett.73.58}
Reck, M., Zeilinger, A., Bernstein, H.~J., and Bertani, P.
\newblock Experimental realization of any discrete unitary operator.
\newblock \emph{Phys. Rev. Lett.}, 73:\penalty0 58--61, Jul 1994.
\newblock \doi{10.1103/PhysRevLett.73.58}.
\newblock URL \url{https://link.aps.org/doi/10.1103/PhysRevLett.73.58}.

\bibitem[Rezende et~al.(2019)Rezende, Racanière, Higgins, and
  Toth]{rezende2019equivariant}
Rezende, D.~J., Racanière, S., Higgins, I., and Toth, P.
\newblock Equivariant hamiltonian flows, 2019.

\bibitem[Sakurai \& Napolitano(2017)Sakurai and Napolitano]{QM_book}
Sakurai, J. and Napolitano, J.
\newblock \emph{Modern Quantum Mechanics}.
\newblock Cambridge University Press, 2017.
\newblock ISBN 9781108422413.
\newblock URL \url{https://books.google.nl/books?id=010yDwAAQBAJ}.

\bibitem[Shen et~al.(2017)Shen, Harris, Skirlo, Prabhu, Baehr-Jones, Hochberg,
  Sun, Zhao, Larochelle, Englund, et~al.]{shen2017deep}
Shen, Y., Harris, N.~C., Skirlo, S., Prabhu, M., Baehr-Jones, T., Hochberg, M.,
  Sun, X., Zhao, S., Larochelle, H., Englund, D., et~al.
\newblock Deep learning with coherent nanophotonic circuits.
\newblock \emph{Nature Photonics}, 11\penalty0 (7):\penalty0 441, 2017.

\bibitem[Simon et~al.(1987)Simon, Sudarshan, and Mukunda]{PhysRevA.36.3868}
Simon, R., Sudarshan, E. C.~G., and Mukunda, N.
\newblock Gaussian-wigner distributions in quantum mechanics and optics.
\newblock \emph{Phys. Rev. A}, 36:\penalty0 3868--3880, Oct 1987.
\newblock \doi{10.1103/PhysRevA.36.3868}.
\newblock URL \url{https://link.aps.org/doi/10.1103/PhysRevA.36.3868}.

\bibitem[Steinbrecher et~al.(2019)Steinbrecher, Olson, Englund, and
  Carolan]{steinbrecher2019quantum}
Steinbrecher, G.~R., Olson, J.~P., Englund, D., and Carolan, J.
\newblock Quantum optical neural networks.
\newblock \emph{npj Quantum Information}, 5\penalty0 (1):\penalty0 1--9, 2019.

\bibitem[Toth et~al.(2020)Toth, Rezende, Jaegle, Racanière, Botev, and
  Higgins]{toth2020hamiltonian}
Toth, P., Rezende, D.~J., Jaegle, A., Racanière, S., Botev, A., and Higgins,
  I.
\newblock Hamiltonian generative networks, 2020.

\bibitem[{Verdon} et~al.(2018){Verdon}, {Pye}, and {Broughton}]{baqprop}
{Verdon}, G., {Pye}, J., and {Broughton}, M.
\newblock {A Universal Training Algorithm for Quantum Deep Learning}.
\newblock \emph{arXiv e-prints}, art. arXiv:1806.09729, June 2018.

\bibitem[Wang et~al.(2016)Wang, Yan, and Oates]{wang2016time}
Wang, Z., Yan, W., and Oates, T.
\newblock Time series classification from scratch with deep neural networks: A
  strong baseline, 2016.

\bibitem[Watrous(2018)]{watrous2018theory}
Watrous, J.
\newblock \emph{The Theory of Quantum Information}.
\newblock Cambridge University Press, 2018.
\newblock ISBN 9781107180567.
\newblock URL \url{https://books.google.nl/books?id=GRNSDwAAQBAJ}.

\bibitem[Weedbrook et~al.(2012)Weedbrook, Pirandola, García-Patrón, Cerf,
  Ralph, Shapiro, and Lloyd]{Weedbrook_2012}
Weedbrook, C., Pirandola, S., García-Patrón, R., Cerf, N.~J., Ralph, T.~C.,
  Shapiro, J.~H., and Lloyd, S.
\newblock Gaussian quantum information.
\newblock \emph{Reviews of Modern Physics}, 84\penalty0 (2):\penalty0
  621–669, May 2012.
\newblock ISSN 1539-0756.
\newblock \doi{10.1103/revmodphys.84.621}.
\newblock URL \url{http://dx.doi.org/10.1103/RevModPhys.84.621}.

\bibitem[Zhao et~al.(2019)Zhao, Fitzsimons, and Fitzsimons]{zhao2019quantum}
Zhao, Z., Fitzsimons, J.~K., and Fitzsimons, J.~F.
\newblock Quantum-assisted gaussian process regression.
\newblock \emph{Physical Review A}, 99\penalty0 (5):\penalty0 052331, 2019.

\end{thebibliography}
